\documentclass{article}

\usepackage{amssymb,amsthm,latexsym, amsmath, amstext, amsfonts}
\usepackage{graphicx,array}
\usepackage[english]{babel}
\usepackage{color}
\usepackage{caption}
\usepackage{subcaption}

\usepackage{booktabs} 
\usepackage{colortbl} 
\usepackage{xcolor} 
\usepackage{xfrac}

\newcommand{\R}{\mathbb{R}}

\newcommand{\Z}{\mathbb{Z}}
\newcommand{\I}{\mathcal{I}}

\newcommand{\x}{\mathcal{X}}

\newtheoremstyle{bsmm}%
{10pt}{10pt}{\itshape}%
{}{\scshape}{.}{ }%
{\thmname{#1}\thmnumber{ (#2)}\thmnote{ (#3)}}

\theoremstyle{bsmm}
\newtheorem{thm}{Theorem}[section]
\newtheorem{prop}[thm]{Proposition}

\theoremstyle{remark}

\newtheorem{remark}[thm]{Remark}

\newtheorem{defi}[thm]{Definition}

\title{Itinerary synchronization in a network of nearly identical PWL systems coupled with unidirectional links and ring topology}
\author{A.~Anzo-Hern\'andez$^a$, E.~Campos-Cant\'on$^{b*}$  and Matthew Nicol$^c$ }
\date{}

\begin{document}
\maketitle
\begin{center}
$^a${\small 
C\'atedras CONACYT - Benem\'erita Universidad Aut\'onoma de Puebla,\\
\textsc{Facultad de Ciencias F\'isico-Matem\'aticas,\\
Benem\'erita Universidad Aut\'onoma de Puebla,\\
Avenida San Claudio y 18 Sur, Colonia San Manuel, 72570.\\
Puebla, Puebla, M\'exico.} {\it andres.anzo@hotmail.com}
\vskip 2ex}
$^{b}${\small \textsc{Divisi\'on de Matem\'aticas Aplicadas,}\\
 {Instituto Potosino de Investigaci\'on Cient\'{\i}fica y Tecnol\'ogica A.C.}\\
 \textsc{Camino a la Presa San Jos\'e 2055 col. Lomas 4a Secci\'on, 78216,
 San Luis Potos\'{\i}, SLP, M\'{e}xico.} {\it eric.campos@ipicyt.edu.mx,}{\normalsize $^*$}Corresponding author. 
\vskip 2ex}
$^c${\small \textsc{Mathematics Department, University of Houston},\\
{Houston, Texas},\\
\textsc{77204-3008, USA.
{\it nicol@math.uh.edu.}
\vskip 2ex}}
\end{center}
\setcounter{page}{1}

\begin{abstract}
In this paper the collective dynamics of $N$-coupled piecewise linear (PWL) systems with different number of scrolls and coupled in a  master-slave sequence configuration is studied, \textit{i.e.} a ring connection with unidirectional links. Itinerary synchronization is proposed to detect synchrony behavior with systems that can present generalized multistability. Itinerary synchronization consists in analyzing the symbolic dynamics of the systems by assigning different numbers to the regions where the scrolls are generated. It is shown that in certain parameter regimes  if the inner connection between nodes is given by means of considering all the state variables of the system, then itinerary synchronization occurs and the coordinate motion is determined by the node with the smallest number of scrolls. Thus the collective behavior in all the nodes of the network is determined by the node with least scrolls in its attractor leading to generalized multistability phenomena which can be detected via itinerary synchronization.  Results about attacks to the network are also presented, for example,  when the PWL system is attacked by removing a given link  to produce an open ring configuration.  Depending on the inner connection properties, the nodes present multistability or preservation in the number of scrolls of the attractors. 
\end{abstract}

\noindent {\bf keywords:}
Piecewise linear systems; chaos; dynamical networks, multiscroll attractor.

\section{Introduction}

Piecewise linear (PWL) systems are used to construct simple  chaotic oscillators capable of generating various multiscroll attractors in the phase space. These systems contain a linear part plus a nonlinear element characterized by a switching law. One of the most studied PWL system is the so called Chua's circuit, whose nonlinear part (also named the Chua diode) generates two scroll attractors \cite{Chua1992,Kennedy1993}. Inspired by  the Chua circuit, a great number of PWL systems have been produced  via various switching systems~\cite{Liu2012}. A review and summary of  different approaches to generate multiscroll attractors can be found in~\cite{YALCIN2002,Lue2006, Campos2016} and references therein.

Synchronization phenomena in a pair of coupled PWL systems has also attracted  attention in the context of nonlinear dynamical systems theory and its applications \cite{Gamez-Guzman2009, Munoz-Pacheco2012}. 

In general, we say that a set of dynamical systems achieve synchronization if they adjust their motion to approach a  common trajectory (in some sense) by means of  interactions~\cite{Boccaletti2006b}.

 One way  to study synchronization in a pair of PWL systems is to couple them in a master-slave configuration \cite{Chua1992, Fuhong2012}. In \cite{Fuhong2015} the dynamics mechanism of the projective synchronization of Chua circuits with different scrolls is investigated. 
In~\cite{Jimenez2013},  a master-slave system composed of PWL systems is considered in which the slave system displays more scrolls in its attractor than the master system. The main result is that the slave system synchronizes with the master system by reducing its number of attractor scrolls, while the master preserves its number of scrolls. A consequence is the emergence of multistability phenomena. For instance, if the number of scrolls presented by the master system is less than the number of scrolls presented by the slave system, then the slave system can oscillate in multiple basins of attraction depending on its initial condition. Conversely, when  the system of~\cite{Jimenez2013} is adjusted so that the master system displays more scrolls than the slave system when uncoupled then the slave system increases its number of attractor  scrolls when coupled.

 We study  a system composed of an ensemble of master-slave systems coupled in a  ring configuration network; {\it i.e.}, a dynamical network where each node is a PWL-system with different number of scrolls in the attractors and connected in a ring topology with directional links. In order to address this problem, we introduce three concepts: 1)  scroll-degree, which is defined as the number of scrolls of an attractor in a given node; 2)  a network of nearly identical nodes, \textit{i.e.}, a dynamical network composed of PWL systems with perhaps different scroll degree but similar form and 3) itinerary synchronization based on symbolic dynamics. A PWL system is defined by means of a partition  of the space where linear systems act, so this natural partition is useful for analyzing  synchronization between dynamical systems by using symbolic dynamics. Of course itinerary synchronization does not imply  complete synchronization, where trajectories converge to a single one. In this paper we  study the emergence of itinerary synchronization, multistability and the preservation of the scroll number of a network of nearly identical nodes.

In real networks, there are common types of network attacks. There are attacks where information is monitored, this attack is known as passive attack; other attacks alter the information with intent to corrupt it or destroy the network itself, this kind of attack is known as active attack. In this work,  link attacks are introduced to the ring topology in order to study the effect of topology changes in the collective dynamics of the network. Such attacks modify the topology by deleting a single link, transforming the structure to an open chain of coupled systems which we call an open ring. We have formulated two possible scenarios after the link attack: a) the first node in the chain has the largest scroll-degree or, b) it has the smallest one. In both scenarios we assume that the inner coupling matrix is the identity matrix \textit{i.e.} the coupling between any pair of nodes is throughout all its state variables.

We believe that a study of  such collective phenomena in a network of PWL systems contributes to a better understanding of the collective dynamics of a network with non identical nodes. The general case  is difficult to tackle. For example, in~\cite{Sun2009}, Sun \textit{et.al.} studied  the case in which nodes are nearly identical in the sense that each node has a slight parametric mismatch. The authors proposed an extension of the master stability functions in this type of dynamical network. On the other hand, based on stability analysis with multiple Lyapunov functions, Zhao \textit{et.al.} established synchronization criteria for certain  networks of non-identical nodes with the same equilibria point \cite{Zhao2011}. The authors proposed stability conditions in terms of inequalities involving matrix spectra which are,  computationally speaking, difficult to solve. To the best of our knowledge, multistability and scroll-degree preservation have not been studied  in the context of dynamical networks.
This is a problem with potential applications.  For example, a dynamical network of PWL system increases the complexity in its behavior, so this can be used to generate a secure communications system, and due to the simplicity, it could be easily implemented in an electronic device.

We have organized this paper as follows: In section 2 we introduce some mathematical preliminaries. In section 3 the dynamics of N-coupled PWL systems in a ring topology network is analyzed. In section 4 some examples about itinerary synchronization are studied and different forms of couplings are also considered. Moreover, in this section we study the dynamics of a network when some link attacks occur.  Finally, in section 5 we consider open problems.

\section{Mathematical Preliminaries}

 \subsection{Piecewise linear dynamical systems}

Let $T:X\to X$, with $X\subset \R^n$ and $n \in \Z^{+}$, be a piecewise linear dynamical system whose dynamics is given by a family of sub-systems  of the form 

\begin{equation}\label{eq:family_eq}
\dot{\x} = A_{\tau}\x + B_{\tau},
\end{equation}

\noindent where  $\x= (x_{1},\ldots,x_{n})^T \in \R^n$ is the state vector, $A_{\tau}  = \{ \alpha^{\tau}_{ij} \} \in \R^{n \times n}$, with $\alpha^{\tau}_{ij} \in \R^{+}$, and  $B_{\tau} = (\beta_{\tau 1},\ldots,\beta_{\tau n})^{T}  \in \R^n$  are the linear operators and constant real vectors of the $\tau$th-subsystems, respectively. The index $\tau \in \I=\{1,\ldots,\eta\}$ is given by a rule that switches the activation of a sub-system in order to determine the dynamics of the PWL system. 

Note that the index $\tau$ generates the symbolic dynamics of activation of the subsystems, so depending on the initial condition the itinerary is given. The selection of the index $\tau$ can be given according to a predefined itinerary and controlling by time; or  assuming that $\tau$ takes its value according to the state variable $\chi$ and a finite partition of the state-space $\mathcal{P}=\{P_1,\ldots,P_r\}$, with $r\in \Z^+$. 
\begin{defi} 
Let $X$ be a subset of $\R^n$ and $\mathcal{P}=\{P_1,\ldots,P_r\}$ ($r>1$) be a finite partition of $X$, that is, $X=\bigcup_{1\leq i\leq r} P_i$, and $P_i\cap P_j= \emptyset$ for $i\neq j$. Each element of the set $\mathcal{P}$ is called an atom.
\end{defi}

An easy way to generate a partition $\mathcal{P}$ is given by considering a vector $\mathbf{v} \in \R^{n}$ (with $\mathbf{v} \neq 0$) and a set of scalars $\delta_1 < \delta_2 < \cdots < \delta_{\eta} $ such that each $P_{i} = \{ \x \in  \R^{n}:  \delta_{i} \leq \mathbf{v}^{T} \x  < \delta_{i+1} \}$. The hyperplanes  $\mathbf{v}^{T} \x = \delta_{i}$ are called the switching surfaces, with $i=1,\dots,\eta$. Without loss of generality, we assume that the hyperplanes $\mathbf{v}^{T} \x = \delta_{i}$ (for $i = 1,2,\ldots,\eta$) are defined with $\mathbf{v} = (1,0,\ldots,0)^{T} \in \R^{n}$. 

In this paper we consider a piece-wise linear system $(T,\mathcal{P})$, such that its restriction to each atom $P_i$ is $T(\x_i^*)=0$ for one $\x_i^* \in P_i$  with $i \in \I$, \textit{i.e.}, $\x_i^*=-A_{\tau}^{-1}B_{\tau}$. We assume that the switching signal depends on the state variable and is defined as follows:
\begin{defi} 
 Let $\I = \{1,2,\ldots,\eta \}$ be an index set that labels each element of the family of the sub-systems \eqref{eq:family_eq}. A pure-state dependent and piecewise constant function $\kappa: \R^{n} \rightarrow \I =\{1,\ldots,\eta\}$ of the form
 
 \begin{equation}\label{eq:switching_law}
\kappa(\x) =  \left\{
\begin{array}{lll}
      1, & \text{if}         & \x \in P_1;\\
      2, & \text{if}        & \x \in P_2;\\
       \vdots     &             &   \vdots     \\
      \eta, & \text{if}    & \x \in P_\eta;
\end{array} 
\right.
\end{equation} 

\noindent is called a switching signal. Furthermore, if $\kappa(\x) = \tau_{i} \in \I$ is the value of the switching signal during the time interval $t \in [t_{i},t_{i+1})$, then $\mathcal{S}(\x_0)=\{\tau_0,\tau_1,\ldots, \tau_m, \ldots\}$ stands for the itinerary generated by $\kappa(\x_0)$ at $\x_0$ and, $\mathcal{S}(i,\x_0)$ is the element $\tau_i \in \mathcal{S}(\x_0)$ that occurs at time $t_i$, this defines a  set $\Delta_t=\{t_0,t_1,\ldots,t_m,\ldots\}$.  
\end{defi}

Note that $\tau$'s changes only when the orbit $\phi(t,\chi_0)$  goes from one atom $P_i$ to another $P_j$, $i\neq j$.  
\begin{defi} 
	A  $\eta$-PWL system is composed of two sets: $\textbf{A} = \{ A_{1}, \ldots, A_{\eta} \}$ and $\textbf{B} = \{ B_{1},B_{2},\ldots,B_{\eta}  \}$, with $A_{\tau} = \{ \alpha^{\tau}_{ij} \} \in \R^{n \times n}$ ($\alpha^{\tau}_{ij} \in \R$) and $B_{\tau} = (\beta_{\tau 1},\ldots,\beta_{\tau n})^{T}  \in \R^{n}$; and a pure-feedback switching signal $\kappa: \R^{n} \rightarrow \I = \{1,2,\ldots,\eta \}$ so that:

\begin{equation}\label{eq:PWL}
\dot{\x} =   \left\{
\begin{array}{lll}
      A_{1}\x + B_{1},          & if  & \kappa(\x) =  1; \\
      A_{2}\x + B_{2},          & if  & \kappa(\x) =  2; \\
      \hspace{25pt} \vdots  &     &  \hspace{25pt}  \vdots     \\
      A_{\eta}\x + B_{\eta},  & if   & \kappa(\x) = \eta.  
\end{array} 
\right.
\end{equation}
\end{defi}
 
We can rewrite \eqref{eq:PWL} in a more compact form as:
\begin{equation}\label{eq:PWL2}
  \dot{\x} = A_{\kappa(\x)} \x + B_{\kappa (\x)}. 
\end{equation}
\begin{defi} 
Two  $\eta_1$-PWL and $\eta_2$-PWL systems are called quasi-sym\-met\-ri\-cal if they are governed by the same linear operator  $A=A_i$ for all $i$  but  $\eta_1 \neq \eta_2$.
\end{defi}
In particular, we assume that the dimension $n=3$ and that the eigenspectra of linear operators $A_{\tau}\in \R^{3\times 3}$ have the following features: a) at least one eigenvalue is a real number; and 2) at least two eigenvalues are complex numbers. There is an approach to generate dynamical systems based on these linear dissipative systems (sometimes called  an unstable  dissipative system (UDS) \cite{Ontanon-Garcia2014}). In this paper we use a particular type of UDS called \textit{Type I}:
\begin{defi} 
	A subsystem $(A_\tau, B_\tau)$ of the system \eqref{eq:PWL2} in $\R^3$ is said to be an UDS of  \textit{Type I} if the eigenvalues of the linear operator $A_{\tau}$ denoted by $\lambda_{i}$ satisfy: $\sum^{3}_{i=1}  \lambda_{i}  < 0$; such that  $\lambda_{1}$  is a negative real eigenvalue and; the other two $\lambda_{2}$ and $\lambda_{3}$  are complex conjugate eigenvalues with positive real part. The system is an UDS of \textit{Type II} if one $\lambda_{i}$ is a positive real eigenvalue and; the other two $\lambda_i$ are complex conjugate eigenvalues with negative real part.
\end{defi}

So, to each value $\kappa(\x) = \tau \in \I$, is associated an atom $P_{\tau} \subset \R^{n}$, containing an equilibrium point $\chi_{\tau}^{*} = -A_{}^{-1}B_{\tau}$ which has a stable manifold $E^{s} = Span \{ \bar{v}_{j} \in \R^{3}: \alpha_{j} < 0\}$ and an unstable manifold $E^{u}  = Span \{ \bar{v}_{j} \in \R^{3}: \alpha_{j} > 0\}$, with $\bar{v}_{j}$ an eigenvector of the linear operator $A$ and $\lambda_{j}= \alpha_{j} + i\beta_{j}$ its corresponding eigenvalue; \textit{i.e.} it is a saddle equilibrium point. We are interested in bounded flows which are generated by quasi-symmetrical $\eta$-PWL systems such that for any initial condition $\x_{0} \in \R^{3}$, the orbit $\phi(t,\chi_{0})$ of the $\eta$-PWL system \eqref{eq:PWL2} is trapped in a one-spiral trajectory in the atom $P_{\tau}$ called a scroll. The orbit escapes from one atom to other due to the unstable manifold in each atom.  In this context, the  system $\eta$-PWL  \eqref{eq:PWL2} can display various multi-scroll attractors as a result of a combination of several unstable one-spiral trajectories, while  the switching between regions is governed by the function \eqref{eq:switching_law}. 

\begin{defi} 
	 The scroll-degree of a $\eta$-PWL system \eqref{eq:PWL2} based on UDS {\it Type I} is the maximum number of scrolls that the PWL system can display in the attractor.
\end{defi}	 
	
In this work we consider the same linear operator $A$, so $A_{\tau}=A$ for all $\tau$.

\begin{figure}[ht] 
\centering
\includegraphics[width=12cm,height=6cm]{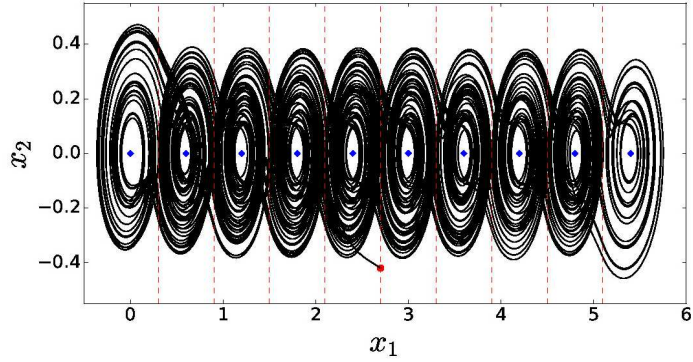}
\caption{Projection of the attractor generated by the quasi-symmetrical 10-PWL(S) system onto the ($x_{1},x_{2}$) plane. The dashed lines mark the division between the atoms.}
\label{fig:Fig1}
\end{figure}

{\bf Example 1:}  
In order to illustrate the generation of multiscroll attractors  using \eqref{eq:PWL2}, we consider a quasi-symmetrical 10-PWL system defined in $\R^3$ with state vector $\x = (x_{1},x_{2},x_{3})^{T}$ and  linear operator defined as follows

\begin{equation}\label{eq:example1}
A = \left( \begin{array}{ccc}
0 & 1 & 0 \\
0 & 0 & 1 \\
-\alpha_{31} & -\alpha_{32} & -\alpha_{33} \end{array} \right); 
\end{equation}
\noindent where $\alpha_{31} = 1.5, \alpha_{32} = 1$ and $\alpha_{33} = 1$;
the set of constants vectors 
$$\textbf{B} = \{ B_{1} = (0,0,0)^{T}, B_{2} = (0,0,0.9)^{T}, B_{3} = (0,0,1.8)^{T}, B_{4} = (0,0,2.7)^{T},$$
$$ B_{5} = (0,0,3.6)^{T}, B_{6} = (0,0,4.5)^{T} ,  B_{7} = (0,0,5.4)^{T}, B_{8} = (0,0,6.3)^{T}, $$
$$B_{9} = (0,0,7.2)^{T}, B_{10} = (0,0,8.1)^{T}\};$$ and the partition: 
\begin{equation}\label{eq:switching_law_10}
\begin{array}{r}
\mathcal{P} \; = \; \{\;	  P_{1} = \{\x \in \mathbf{R}^{3}: x_{1} < 0.3 \},
      P_{2} = \{\x \in \{\x \in \mathbf{R}^{3}: 0.3 \leq x_{1} < 0.9 \},\\
      P_{3} = \{\x \in \mathbf{R}^{3}: 0.9 \leq x_{1} < 1.5 \},
      P_{4} = \{\x \in \mathbf{R}^{3}: 1.5 \leq x_{1} < 2.1 \},\\
      P_{5} = \{\x \in \mathbf{R}^{3}: 2.1 \leq x_{1} < 2.7 \},
      P_{6} = \{\x \in \mathbf{R}^{3}: 2.7 \leq x_{1} < 3.3 \},\\
      P_{7} =\{\x \in \mathbf{R}^{3}: 3.3 \leq x_{1} < 3.9 \},
      P_{8} = \{\x \in \mathbf{R}^{3}: 3.9 \leq x_{1} < 4.5\}\\
      P_{9} = \{\x \in \mathbf{R}^{3}: 4.5 \leq x_{1} < 5.1\}, 
      P_{10} = \{\x \in \mathbf{R}^{3}: x_{1} \geq 5.1\} \}
 \end{array}
 \end{equation}
 
The eigenvalues of $A$ are $\lambda_{1} = -1.20$ and $\lambda_{2,3} = 0.10  \pm 1.11 \textit{i}$. By Definition 2.4, the system is an UDS of \textit{Type I}. The equilibrium points for this system are at $\chi^{*}_{1} =  (0,0,0)^T, $ $\chi^{*}_{2} =  (0.6,0,0)^T$, $ \chi^{*}_{3} =  (1.2,0,0)^T$, $\chi^{*}_{4} =  (1.8,0,0)^T $, $\chi^{*}_{5} =  (2.4,0,0)^T$, $\chi^{*}_{6} =  (3,0,0)^T$, $ \chi^{*}_{7} =  (3.6,0,0)^T$, $ \chi^{*}_{8} =  (4.2,0,0)^T$, $\chi^{*}_{9} =  (4.8,0,0)^T  $ and $\chi^{*}_{10} =  (5.4,0,0)^T$. Figure \eqref{fig:Fig1} depicts the projection of the attractor generated by the quasi-symmetrical 10-PWL(S) system onto the ($x_{1},x_{2}$) plane with initial condition $\chi_{0} = (4.81, -0.38, 0.09)^{T}$. We solved this system~\eqref{eq:PWL}  numerically Runge-Kutta  with $2000$ time iterations and step-size $h=0.01$.

\subsection{Two coupled PWL systems}

Consider two quasi-symmetrical $\eta$-PWL systems defined by \eqref{eq:PWL2}, i.e they have different  scroll-degrees.  We couple in a Master-Slave configuration given as follows.
\begin{equation} \label{eq:MS}
\begin{split}
\dot{\x}_{m} & =  A\x_{m} +  B_{\kappa_{m} ( \x_{m}) },   \\ 
\dot{\x}_{s}  & =  A\x_{s}   +  B_{\kappa_{s} ( \x_{s} ) }  + c\Gamma ( \x_{m} - \x_{s}),
\end{split}
\end{equation}
where  $\x_{m} = (x^{m}_{1},x^{m}_{2},x^{m}_{3})^{T}$ and $\x_{s} = (x^{s}_{1},x^{s}_{2},x^{s}_{3})^{T}$ are the state vectors of the master and slave systems, respectively. $\kappa_{i}: \R^{3} \rightarrow \I_{i} = \{1,2,\ldots,\eta_{i}\}$, with $i=m,s$ and  $\eta_{m} \neq \eta_{s}$, is the pure-state-feedback signal of the master system ($i=m$) and slave system ($i=s$). The itineraries generated by $\tau$'s  of the master and slave systems are $\mathcal{S}_m(\x_{m0})=\{\tau_0,\tau_1,\ldots\}$ and $\mathcal{S}_s(\x_{s0})=\{\tau'_0,\tau'_1,\ldots\}$, respectively. The corresponding time sets are given by $\Delta_{tm}=\{t_0,t_1,\ldots \}$ and $\Delta_{ts}=\{t'_0,t'_1,\ldots \}$. The constant matrix $\Gamma  = diag\{r_1,r_2,r_3 \} \in \R^{3 \times 3}$ is the inner linking matrix where $r_{l} = 1$ (for $l = 1,2,3$) if both master and slave systems are linked through their $l$-th state variable, and $r_{l} = 0$ otherwise, and $0< c \in \R$ is the coupling strength.

There are several definitions of synchronization~\cite{Pikovsky2001, Jun2011}, for instance, complete synchronization is given as follows: 
\begin{defi} 
	The master-slave system \eqref{eq:PWL2} is said to achieve complete synchronization if 
	\begin{equation}
	 \lim_{t\to \infty}|| \phi(t,\x_{m0}) - \phi(t,\x_{s0}) || \to 0.
	\end{equation}
	for  $\x_{m0} \neq  \x_{s0}$.
\end{defi}

The symbol $|| \cdot||$ denotes the Euclidean distance in $\R^{3}$.  This mode of synchronization is very strong.  There are weaker and more
generalized notions of  synchronization~\cite{Rulkov1995}.

It has been reported in \cite{Jimenez2013} that in  the type of configuration given by \eqref{eq:MS} the  master system determines the scroll-degree in the slave system. In particular, if  $\eta_{m} < \eta_{s}$, then the master-slave system achieves complete synchronization and different basins of attraction appear equal to $\eta_{s} - \eta_{m}+1$. The trajectories of the slave system  depend on their initial condition. That is, the master-slave configuration results in multiple basins of attraction for the slave. Such a phenomena is called multistability~\cite{Campos-Canton2012}. On the other hand, if $\eta_{m} > \eta_{s}$, then the slave system increases its scroll-degree till it matches the master's scroll-degree.

In order to illustrate the dynamical behavior of the master-slave system, consider two quasi-symmetrical $\eta$-PWL system with linear operator $A$ and a set of constant vectors $\mathbf{B} = \{B_{1}, B_{2}, \ldots, B_{10} \}$ defined in Example 1 (Eq. \eqref{eq:example1}). 
\begin{figure}[ht] 
\centering
\includegraphics[width=11.5cm,height=10cm]{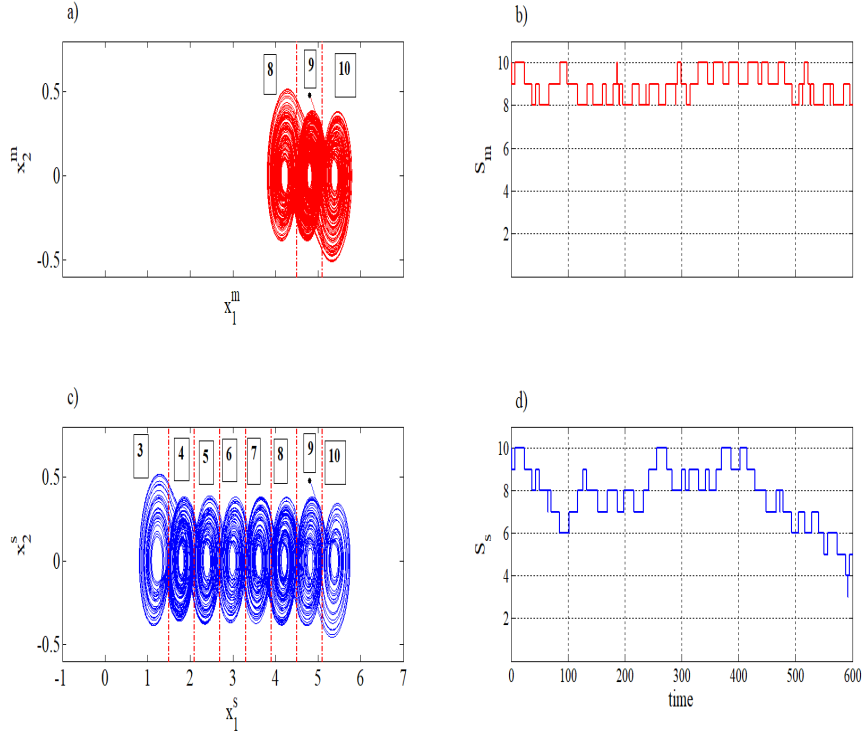}
\caption{\small a) Projection of the master system onto the plane $(x_1^m,x_2^m)$ with initial condition $\chi_{mo}=(4.8,0.48,-0.29)^T$; b)  The master itinerary $S_m(\chi_{m0})$; c) Projection of the slave system onto the plane $(x_1^s,x_2^s)$ with initial condition $\chi_{so}=(4.8,0.48,-0.29)^T$, for coupling strength $c=0$; d) The slave itinerary $S_s(\chi_{s0})$.}
\label{fig:ItiSalveC0}
\end{figure}

{\bf Example 2:} 
As a first example of a coupled pair of multiscroll chaotic systems, suppose that the master's scroll-degree is $\eta_{m} = 3$  and  the slave's scroll-degree is $\eta_{s} = 8$, and both are connected with a coupling strength $c$ and an inner coupling matrix given by $\Gamma = \{0,1,0\}$. The pure-state-feedback signal for the master system $\kappa_{m}: \R^{3} \rightarrow \I_{m}= \{8, 9, 10 \}$ is 
\begin{equation}\label{eq:switching_law_3}
\kappa_{m}(\x) =  \left\{
\begin{array}{lll}
      10, & \text{ if}   & \x \in P_{10} = \{\x \in \R^{3}: x_{1} \geq 5.1\}; \\
      9, & \text{ if}   & \x \in P_{9} = \{\x \in \R^{3}: 4.5 \leq x_{1} < 5.1\};  \\
      8, & \text{ if}   & \x \in P_{8} = \{\x \in \R^{3}:  x_{1} < 4.5\}.   \\
\end{array} 
\right.
\end{equation}

And for the slave system the function $\kappa_{s}: \R^{3} \rightarrow \I_{s} = \{3, 4, \ldots, 10 \}$ is
\begin{equation}\label{eq:switching_law_8}
\kappa_{s}(\x) =  \left\{
\begin{array}{lll}
      10, & \text{ if}          & \x \in P_{10} = \{\x \in \R^{3}: x_{1} \geq 5.1\}; \\
      9, & \text{ if}          & \x \in P_{9} = \{\x \in \R^{3}: 4.5 \leq x_{1} < 5.1\};  \\
      8, & \text{ if}          & \x \in P_{8} = \{\x \in \R^{3}: 3.9 \leq x_{1} < 4.5\};   \\
      7, & \text{ if}          & \x \in P_{7} = \{\x \in \R^{3}: 3.3 \leq x_{1} < 3.9 \};   \\
      6, & \text{ if}          & \x \in P_{6} = \{\x \in \R^{3}: 2.7 \leq x_{1} < 3.3 \};  \\
      5, & \text{ if}          & \x \in P_{5} = \{\x \in \R^{3}: 2.1 \leq x_{1} < 2.7 \};  \\
      4, & \text{ if}          & \x \in P_{4} = \{\x \in \R^{3}: 1.5 \leq x_{1} < 2.1 \};  \\
      3, & \text{ if}          & \x \in P_{3} = \{\x \in \R^{3}:  x_{1} < 1.5 \}.  \\
\end{array} 
\right.
\end{equation}

Using Runge-Kutta  with $200000$ time iterations and a step-size of $h=0.01$, we solve numerically the system \eqref{eq:MS}. 
 Firstly, we analyze the particular case when the coupling strength is $c=0$, {\it i.e.}, the systems are not coupled. Projections of the attractors onto the planes $(x_1^m,x_2^m)$ and $(x_1^s,x_2^s)$ are given in Figures \ref{fig:ItiSalveC0} a) and c), in both cases the master and slave systems start at the same initial condition $\chi_{m0}=\chi_{s0}=(4.8,0.48,-0.29)^T$. This initial condition is indicated with a black dot in figures. However master and slave systems oscillate in a different way due to they have different scroll degree $\eta_m=3$ and $\eta_s=8$. The elements of the index sets $\I_m=\{8,9,10\}$ and  $\I_s=\{3,4,5,6,7,8,9,10\}$ for the master and slave systems, respectively, are indicated on the top of Figures \ref{fig:ItiSalveC0} a) and c).

 Figures \ref{fig:ItiSalveC0} b) and d) show the itineraries $S_m(\chi_{m0})$ and $S_s(\chi_{s0})$ of the master and slave systems, respectively. Note that they are different because the systems have different scroll-degree, even though they start at the same initial condition. The itineraries $S_m(\chi_{m0})$ and $S_s(\chi_{s0})$ are given by the dynamics of the master and slave systems and correspond to the activation of the systems in different atoms of the partitions, {\it i.e.}, the itinerary $S_m(\chi_{m0})$  generated by $\kappa_m:\R^3 \to \I_m$ only takes three values $\{8,9,10\}$, meanwhile the itinerary  $S_s(\chi_{s0})$ generated by $\kappa_s:\R^3 \to \I_s$ takes eight values   $\{3,4,5,6,7,8,9,10\}$.

\begin{figure}[ht] 
\centering
\includegraphics[width=11.5cm,height=10cm]{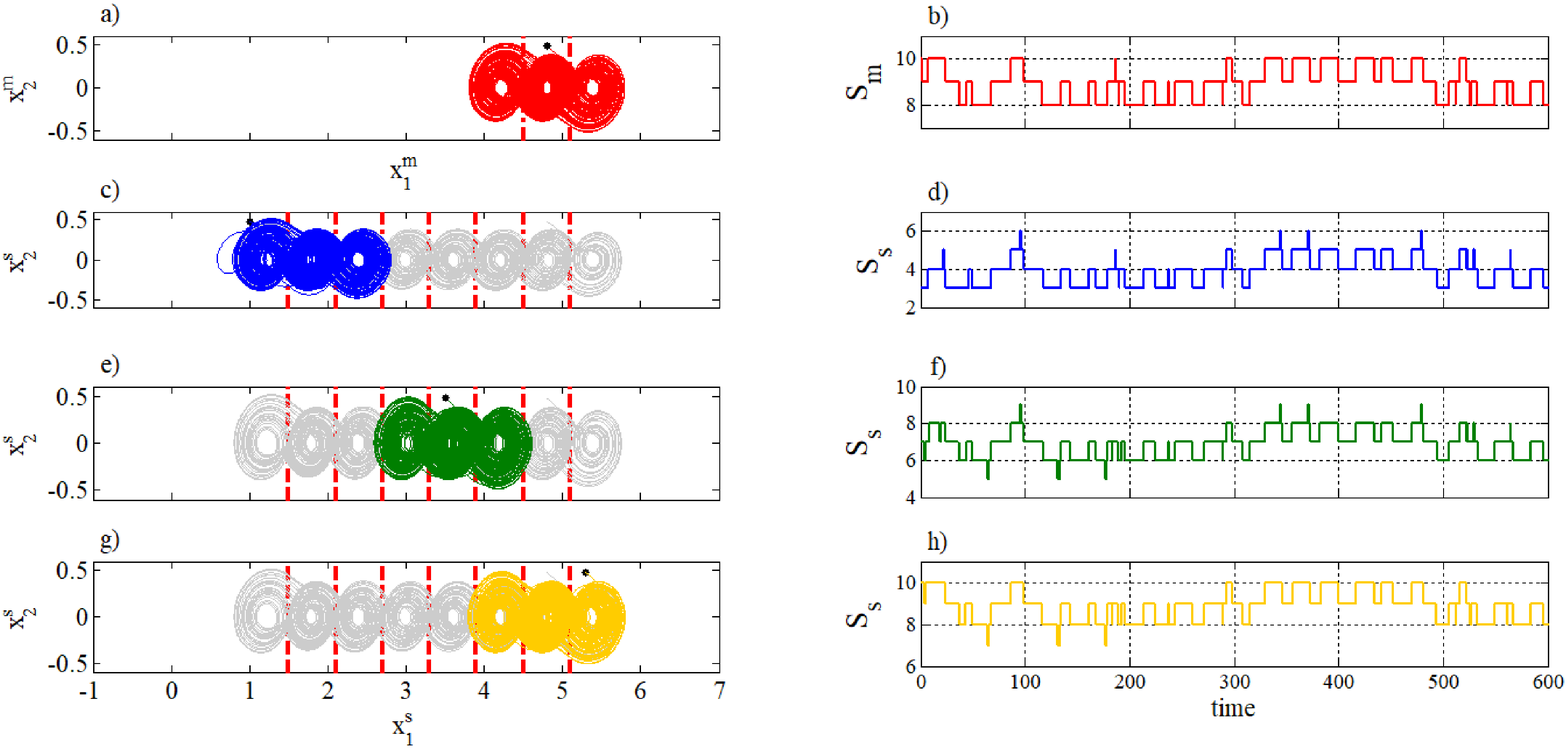}
\caption{\small Projections of the  master-slave system onto the ($x_{1},x_{2}$) plane, for $\eta_{m}=3$, $\eta_{s}=8$, $\Gamma = \{0,1,0\}$ and coupling strength $c = 10$. a) Master system with initial condition $\chi_{mo} = (4.8, 0.48, -0.29)^{T}$ and b) its itinerary $\mathcal{S}_m(\x_{mo})$. c) Slave system with $\chi_{so1} = (1.01, 0.48, -0.29)^{T}$, and d) its itinerary $\mathcal{S}_s(\x_{so1})$.  e) Slave system with  $\chi_{so2} = (3.5, 0.48, -0.29)^{T}$ and f) its itinerary $\mathcal{S}_s(\x_{so2})$. g) Slave system with $\chi_{so3} = (5.3, 0.48, -0.29)^{T}$ and h) its itinerary $\mathcal{S}_s(\x_{so3})$.}
\label{fig:DiffIniCons}
\end{figure}

Now, we set the coupling strength $c=10$ and used different initial conditions for the slave system. 

Figure \ref{fig:DiffIniCons} shows the projections of master-slave system given by \eqref{eq:MS} onto the planes ($x_{1}^m,x_{2}^m$) and ($x_{1}^s,x_{2}^s$).
 Different initial conditions are used for the slave system located at distinct atoms. For the master system the initial condition is $\chi_{mo} = (4.8, 0.48, -0.29)^{T}$, see Figure \ref{fig:DiffIniCons}~a). In specific we use different initial conditions for the slave system $\chi_{so1} = (1.01, 0.48, -0.29)^{T}$ for Figure \ref{fig:DiffIniCons}~c), $\chi_{so2} = (3.5, 0.48, -0.29)^{T}$ for Figure \ref{fig:DiffIniCons}~e) and $\chi_{so3} = (5.3, 0.48, -0.29)^{T}$ for Figure \ref{fig:DiffIniCons}~g).   It is worth to observe that the slave system reduces its scroll-degree to three and, depending on the initial condition, it evolves between distinct basins of attraction, \textit{i.e.}, multistability appears.  We plot in gray  the trajectory of the slave system when it is not coupled with the master system in order to compare it when it is coupled, see Figure \ref{fig:DiffIniCons} c), e) and g).


 Notice that the  itinerary of the master system  $S_m(\chi_{m0})$ generated by $\kappa_m:\R^3 \to \I_m=\{8,9,10\}$ remains, however the  itinerary of the slave system $S_m(\chi_{m0})$ generated by $\kappa_s:\R^3 \to \I_s$ is determined by its initial condition, for instance, the itinerary takes different values according to the atom where the initial condition belongs $\chi_{s0}\in P_i$, for $i=3,\ldots,10$. Now the itinerary of the slave system is restricted to take a subset of the index set $\I_s$, {\it i.e.}, $\I_s(\chi_{s0})\subset  \I_s$ which will be called restricted index set. This is because the number of scrolls that the slave system coupled with $c=10$ displays less scrolls that when it is not coupled. Thus the restricted index sets have different cardinality that is determined by the initial condition  $\chi_{s0}\in P_i$, for $i=3,\ldots,10$. So for these three initial conditions there are three different restricted index sets given as follows:
\begin{equation*}
\kappa_s:\R^3 \to \I_s(\chi_{s0})\subset \I_s=
\left\{
\begin{array}{l}
\I_s(\chi_{s01})=\{3,4,5,6\},\\
\I_s(\chi_{s02})=\{5,6,7,8,9\},\\
\I_s(\chi_{s03})=\{7,8,9,10\}.
\end{array}
\right.
\end{equation*}
The cardinality of the index sets $\I_m, \I_s(\chi_{s01}), \I_s(\chi_{s02})$ and $\I_s(\chi_{s03})$ are $3,4,5,$ and 4, respectively. 

There is a problem if we want to detect similar behaviour under the presence of multistability. The inconvenience is resolved by means of defining a new itinerary based on the trajectory of the systems instead of the dynamics. Then a new partition needs to be defined in the basin of attractions of the systems that can determine an itinerary of the flow of the master and slave systems.

\begin{defi}
Let $\I_B=\{\#_1,\ldots,\#_n\}$ be an index set that labels each element of a partition $P_\phi=\{P'_{1},\ldots,P'_{n}\}$ of the basin of attraction of a dynamical system with flow $\phi$. A function $\kappa:\R^n \to \I_B$ of the form
\begin{equation*}
\kappa(\phi(t,\chi_0))=\left\{
\begin{array}{l}
\#_1, \text{ if } \phi(t,\chi_0)\in P'_1;\\
\#_2, \text{ if } \phi(t,\chi_0)\in P'_2;\\
\vdots\\
\#_n, \text{ if } \phi(t,\chi_0)\in P'_n;
\end{array}
\right.
\end{equation*}
generates the itinerary of the trajectory. If $\kappa (\phi (\chi_0 ) )=s_i \in \I_B$ during the time interval $ t \in [t_i , t_{i+1}) $, then   $S^\phi(\chi_0)=\{s_0,s_1,s_2,\ldots\}$ stands for the itinerary of the trajectory $\phi(\chi_0)$.
\end{defi}  

\begin{remark}
Notice that $S(\chi_0)$ is the itinerary of the dynamics of the system meanwhile $S^\phi(\chi_0)$ is the itinerary of the trajectory of the dynamical system.
\end{remark}

In our setting in order to describe appropriately the flows of a master-slave system via symbolic dynamics it is necessary to consider additional atoms  $P_0,P_{\eta+1}$ at the 
`ends' of the contiguous partition atoms to account for exits and returns to $P_1$ and $P_{\eta}$, respectively, to the partition $\mathcal{P}=\{P_1,\ldots,P_\eta \}$. So we code
according to the partition $  \mathcal{P_\phi}=\{P_{-n},\ldots, P_0,P_1,\ldots,P_\eta,P_{\eta+1}, \ldots, P_{N}\}$. We
obtain a symbolic trajectory by writing down the sequence of symbols corresponding to the successive partition elements visited by the trajectory during a certain period of time. 

For simplicity we generate a new partition $\mathcal{P_\phi}$$=\{P_0,P_1,P_2,\ldots,P_{10},P_{11}\}$ because the flow $\phi(\chi_0)\subset P_\phi$ and the index sets present the same cardinality. 

The partition $\mathcal{P_\phi}$ is given as follow:
\begin{equation}\label{eq:P_phi}
\begin{array}{ll}
\mathcal{P_\phi} & = \; \{\;	  
			P_{0} = \{\x \in \{\x \in \mathbf{R}^{3}: x_{1} < -0.3 \},\\
			&P_{1} = \{\x \in \mathbf{R}^{3}: -0.3 \leq x_{1} < 0.3 \},
			P_{2} = \{\x \in \mathbf{R}^{3}: 0.3 \leq x_{1} < 0.9 \},\\
      &P_{3} = \{\x \in \mathbf{R}^{3}: 0.9 \leq x_{1} < 1.5 \},
      P_{4} = \{\x \in \mathbf{R}^{3}: 1.5 \leq x_{1} < 2.1 \},\\
      &P_{5} = \{\x \in \mathbf{R}^{3}: 2.1 \leq x_{1} < 2.7 \},
      P_{6} = \{\x \in \mathbf{R}^{3}: 2.7 \leq x_{1} < 3.3 \},\\
      &P_{7} =\{\x \in \mathbf{R}^{3}: 3.3 \leq x_{1} < 3.9 \},
      P_{8} = \{\x \in \mathbf{R}^{3}: 3.9 \leq x_{1} < 4.5\},\\
      &P_{9} = \{\x \in \mathbf{R}^{3}: 4.5 \leq x_{1} < 5.1\},
      P_{10} = \{\x \in \mathbf{R}^{3}:5.1\leq x_{1} < 5.7\}, \\
			&P_{11} = \{\x \in \mathbf{R}^{3}: 5.7 \leq x_{1} \} \}.
 \end{array}
 \end{equation}
Thus $\mathcal{S}_m^\phi(\x_{m0})=\{s_0,s_1,\ldots, s_m, \ldots\}$ stands for the itinerary generated by the trajectory of the master system $\phi_m(t,\x_{m0})$ at $\x_{m0}$ and, $\mathcal{S}_m^\phi (i,\x_{m0})$ is the element $s_i \in \mathcal{S}^\phi_m(\x_0)$ that occurs at time $t_i$, so the set $\Delta_{\phi_m}=\{t_0,t_1,\ldots,t_m,\ldots\}$ is generated.  In a similar way, we can define the itinerary, $\mathcal{S}_s^\phi(\x_{s0})$ and the set $\Delta_{\phi_s}=\{t'_0,t'_1,\ldots,t'_m,\ldots\}$ generated by the trajectory of the slave system. We always assume that the initial conditions belong to their respectively basin of attraction of the system.
\begin{figure}[ht] 
\centering
\includegraphics[width=11.5cm,height=10cm]{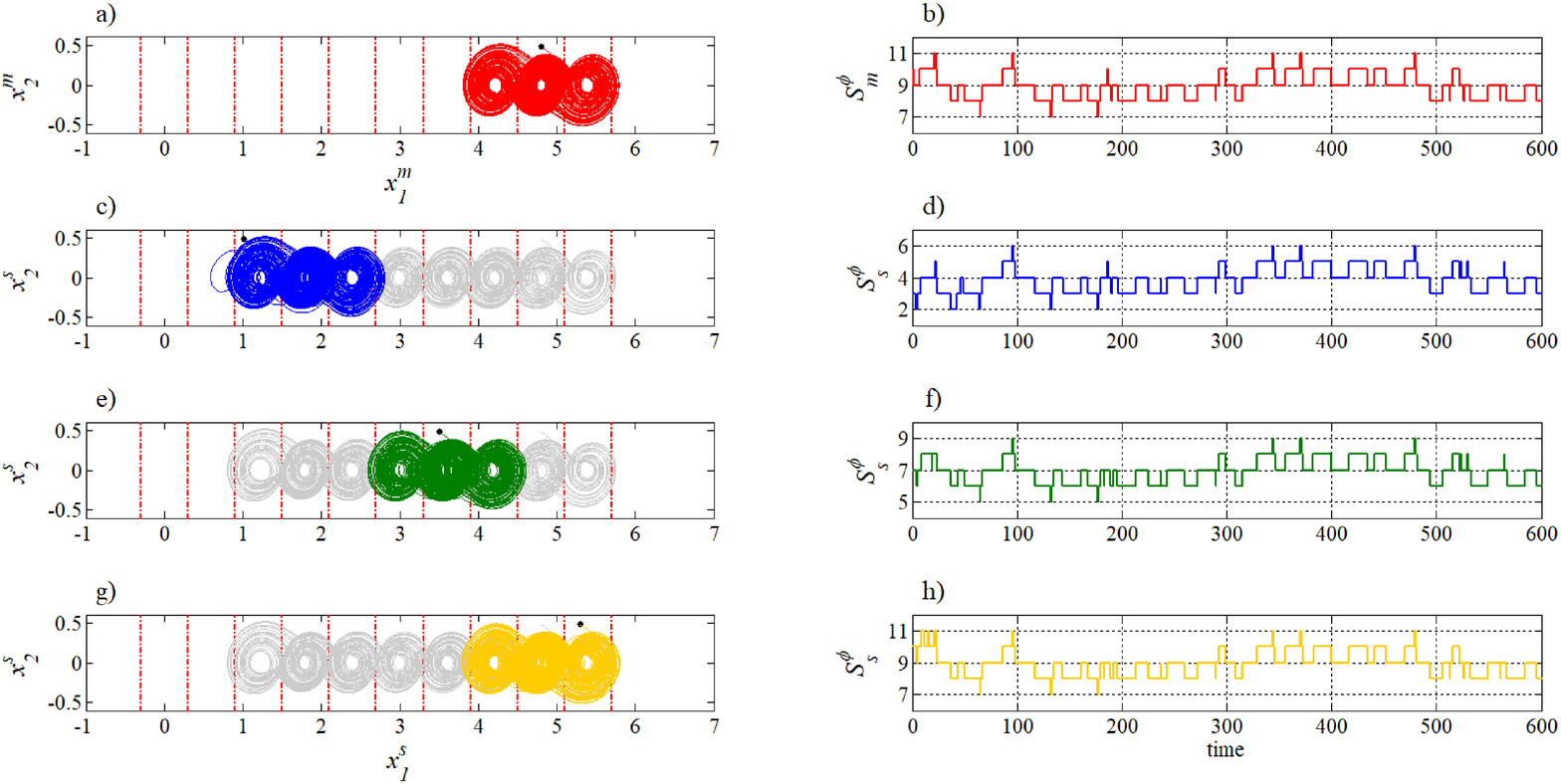}
\caption{\small Projections of the  master-slave system onto the ($x_{1},x_{2}$) plane, for $\eta_{m}=3$, $\eta_{s}=8$, $\Gamma = \{0,1,0\}$ and coupling strength $c = 10$. a) Master system with initial condition $\chi_{mo} = (4.8, 0.48, -0.29)^{T}$ and b) its itinerary $\mathcal{S}_m(\x_{mo})$. c) Slave system with $\chi_{so1} = (1.01, 0.48, -0.29)^{T}$, and d) its itinerary $\mathcal{S}_s(\x_{so1})$.  e) Slave system with  $\chi_{so2} = (3.5, 0.48, -0.29)^{T}$ and f) its itinerary $\mathcal{S}_s(\x_{so2})$. g) Slave system with $\chi_{so3} = (5.3, 0.48, -0.29)^{T}$ and h) its itinerary $\mathcal{S}_s(\x_{so3})$.}
\label{fig:Fig2}
\end{figure}

\begin{figure}[ht] 
\centering
\includegraphics[width=11cm,height=6cm]{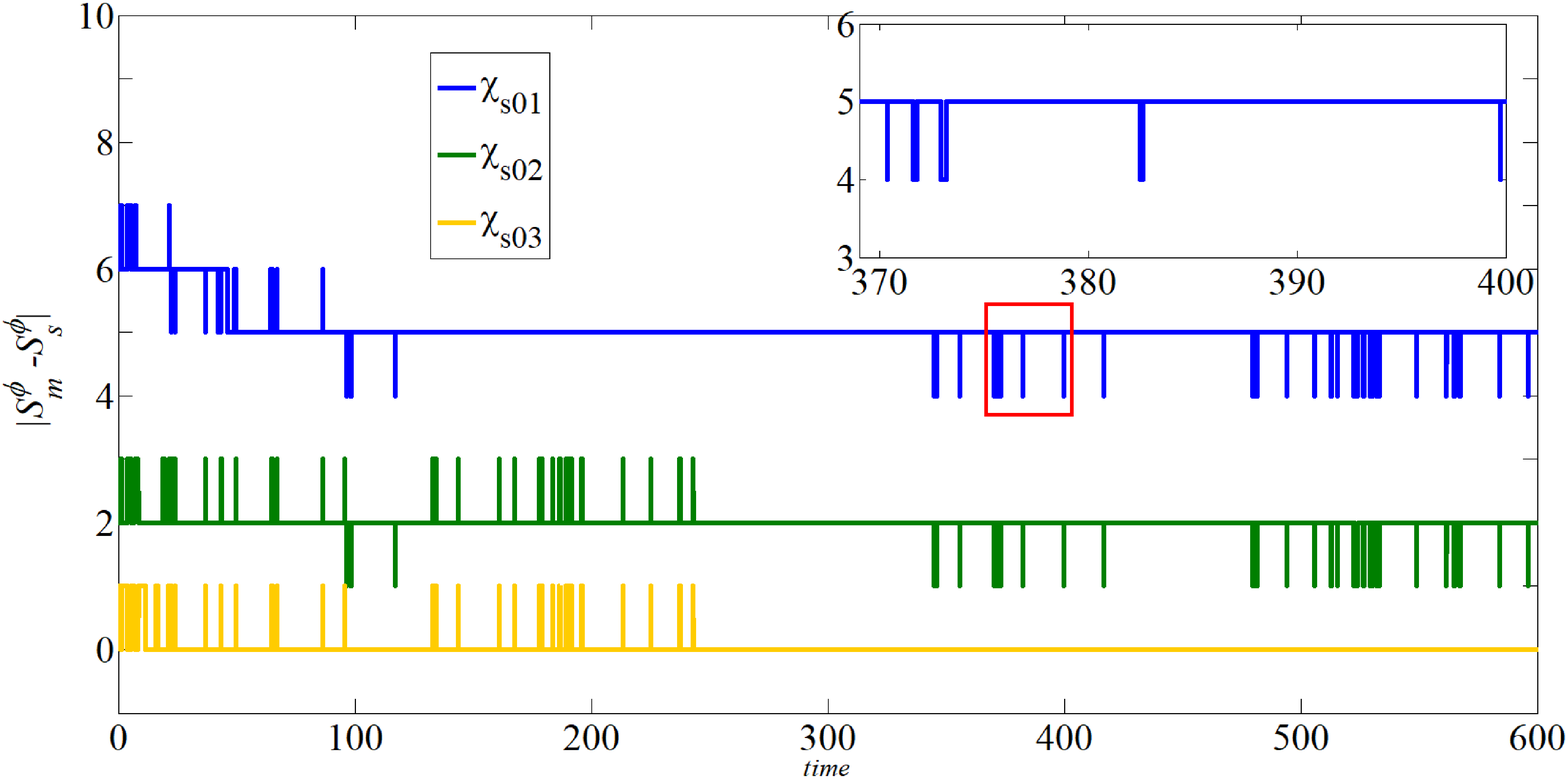}
\caption{\small Difference between the itineraries of the master and the slave systems for the initial conditions given in the example 2. The inner sub-figure shown a zoom of the region market with the red square.}
\label{fig:Fig3}
\end{figure}

Thereafter, the master index set $\I_m$ and restricted index sets $\I_s(\chi_{s01})$, $\I_s(\chi_{s02})$ and $\I_s(\chi_{s03})$ have the same cardinality independently of the initial conditions  $\chi_{s0}\in P_i$, for $i=3,\ldots,10$. Now for these three initial conditions there are three different restricted index sets with the same cardinality given as follows:
\begin{equation}\label{etiqIndxSet}
\kappa_s:\R^3 \to \I_s(\chi_{s0})\subset \I_s=
\left\{
\begin{array}{l}
\I_s(\chi_{s01})=\{2,3,4,5,6\},\\
\I_s(\chi_{s02})=\{5,6,7,8,9\},\\
\I_s(\chi_{s03})=\{7,8,9,10,11\}.
\end{array}
\right.
\end{equation}
And for the master index set:
\begin{equation*}
\kappa_m:\R^3 \to \I_m=\{7,8,9,10,11\}.
\end{equation*}
The cardinality of all of the index sets $\I_m, \I_s(\chi_{s01}), \I_s(\chi_{s02})$ and $\I_s(\chi_{s03})$ is 5. Figure \ref{fig:Fig2}~a) shows the projection of the master attractor onto the plane $(x_1^m,x_2^m)$ and the atoms of $P_\phi$ are marked. Figure~\ref{fig:Fig2}~c),~e) and g) shows the projection of the slave attractor onto the plane $(x_1^s,x_2^s)$ for different initial conditions and the atoms of $P_\phi$ are marked. In Figure \ref{fig:Fig2}~b) we show the itinerary of the master system $\mathcal{S}^\phi_m(\x_{m0}$) and in Figures \ref{fig:Fig2}~d), \ref{fig:Fig2}~f) and \ref{fig:Fig2}~h) the itinerary of the slave system by varying the initial condition. Notice that the itinerary of the trajectory of the master system and the three itineraries of the trajectories of the slave system for different initial conditions visit five different domains. Figure~\ref{fig:Fig3} shows three errors signals which were generated by the difference between the master itinerary $\mathcal{S}^{\phi}_m(i,\chi_{m_0})$ and slave itineraries for different initial conditions $\mathcal{S}^{\phi}_s(i,\chi_{s0})$, with $\chi_{s0}=\{\chi_{s0_1},\chi_{s0_2},\chi_{s0_3}\}$. These signals show spikes that correspond to  when the trajectory goes from one atom to other. But it is possible to estimate  the error  $k$ which is given by a constant in the three cases: for the initial condition $\chi_{so1} = (1.01, 0.48, -0.29)^{T}$ generates on average $k_{1}=5$, $\chi_{so2} = (3.5, 0.48, -0.29)^{T}$ generates $k_{2}=2$ and $\chi_{so3} = (5.3, 0.48, -0.29)^{T}$ generates $k_{3}=0$. 
These small peaks along the error signals show that  we achieve  itinerary synchronization in practice.  As in the inner sub-figure of the Figure \ref{fig:Fig3}  we can observe that  differences due to atom transitions last a small time period.
It is possible to relabel the restricted index sets that correspond to error $k's$ different to zero in order to obtain a error zero. For instance, we have two specific cases that correspond to the restricted index sets $\I_s(\chi_{s0_1})$ and $\I_s(\chi_{s0_2})$, they can relabel as $\I_s(\chi_{s0_3})$.

In the context of synchronization and multistability, we propose the following definition of synchronization based on the itinerary of trajectories in multiscroll attractors:

\begin{defi}  
	The master-slave system \eqref{eq:MS} is said to achieve Itinerary Synchronization if after relabeling the partition atoms 
	\begin{equation}\label{ec_ItineSynch}
	 \lim_{i\to \infty}|\mathcal{S}^{\phi}_m(i,\x_{m0})-\mathcal{S}^{\phi}_s(i,\x_{s0})| = 0,
	\end{equation} 
	for $\x_{m0}\neq \x_{s0}$.
\end{defi}
\begin{remark}
When the master-slave system \eqref{eq:MS} is given by identical systems, {\it i.e.}, with the same scroll-degree, then the coupled system \eqref{eq:MS} can present complete synchronization, {\it i.e.}, the systems are  asymptotically identical. Hence the trajectories of the master and slave visit the same atoms at the same time (without relabelling)  and 
so they have asymptotically  the same itinerary, hence the  itinerary error is zero $|\mathcal{S}^{\phi}_m(i,\x_{m0})-\mathcal{S}^{\phi}_s(i,\x_{s0})|=0$.
\end{remark}
The definition of  Itinerary Synchronization is meant to capture the idea that knowing the itinerary of one sequence determines precisely the itinerary of the other (after relabeling).
\begin{figure}[t] 
\centering
\includegraphics[width=12cm,height=10cm]{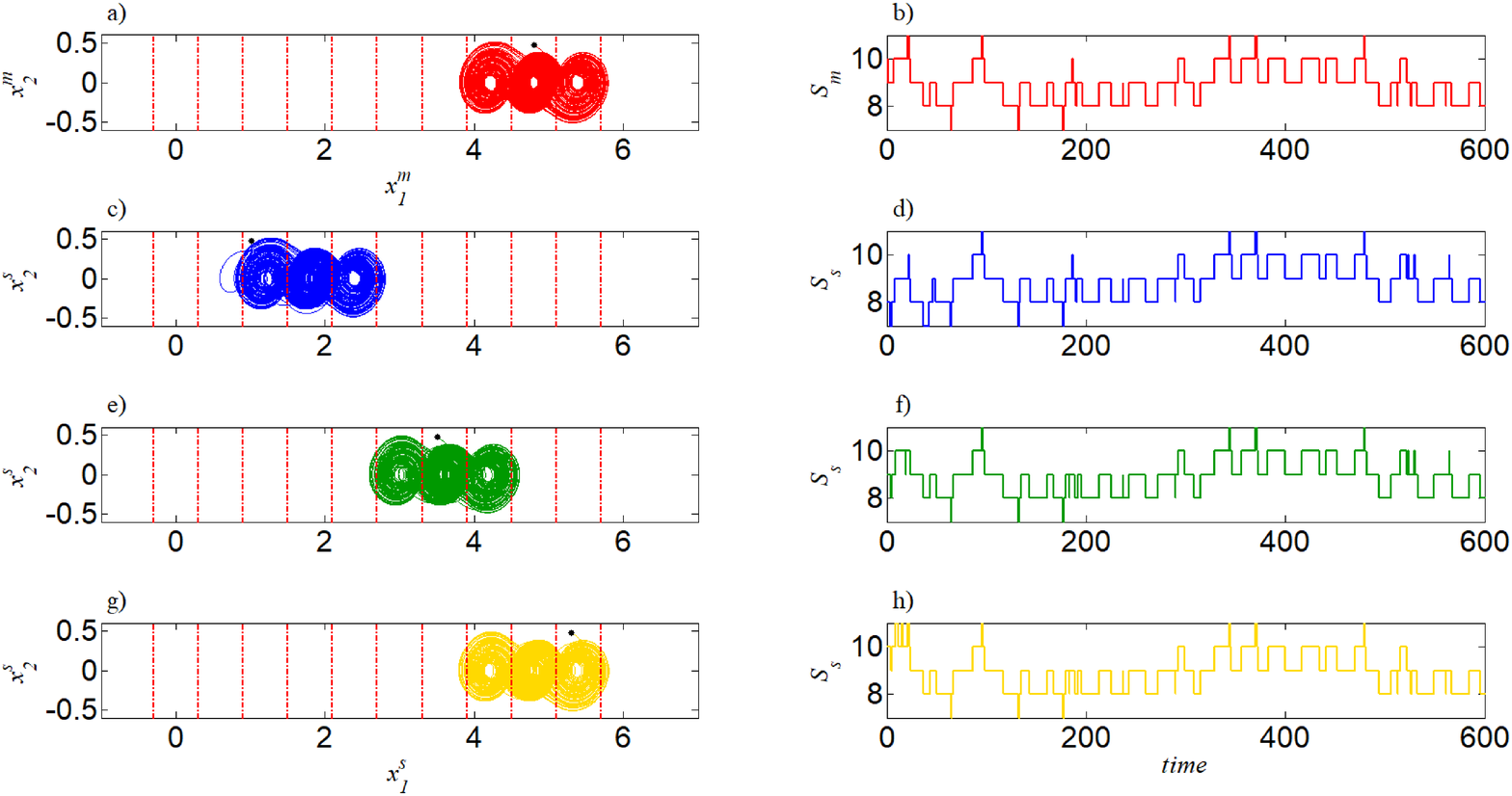}
\caption{\small Projections of the  master-slave system onto the ($x_{1},x_{2}$) plane, for $\eta_{m}=3$, $\eta_{s}=8$, $\Gamma = \{0,1,0\}$ and coupling strength $c = 10$. a) Master system with initial condition $\chi_{mo} = (4.8, 0.48, -0.29)^{T}$ and b) its itinerary $\mathcal{S}_m(\x_{mo})$. c) Slave system with $\chi_{so1} = (1.01, 0.48, -0.29)^{T}$, and d) its itinerary $\mathcal{S}_s(\x_{so1})$ after it was relabeled.  e) Slave system with  $\chi_{so2} = (3.5, 0.48, -0.29)^{T}$ and f) its itinerary $\mathcal{S}_s(\x_{so2})$ after it was relabeled. . g) Slave system with $\chi_{so3} = (5.3, 0.48, -0.29)^{T}$ and h) its itinerary $\mathcal{S}_s(\x_{so3})$ without being relabeled.}
\label{fig:Relabel}
\end{figure}

\begin{figure}[ht] 
\centering
\includegraphics[width=12cm,height=6cm]{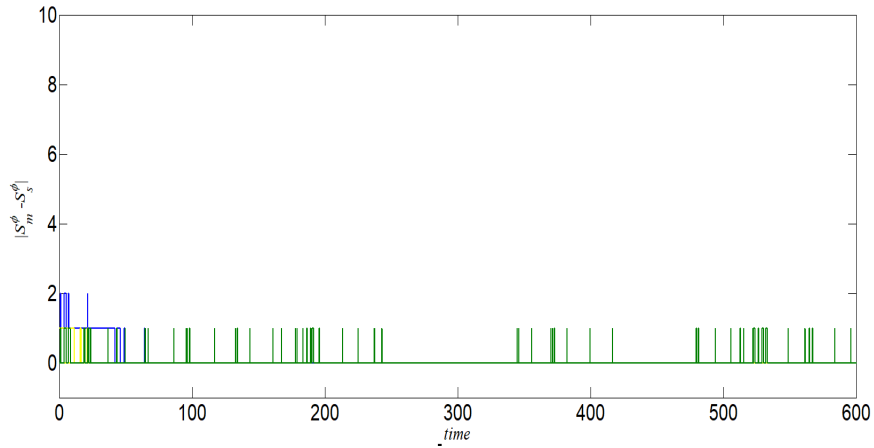}
\caption{\small Difference between the itineraries of the master and the slave systems after relabeling the visited atoms for the initial conditions given in the example 2. }
\label{fig:ErrorRelabel}
\end{figure}

In Figure \ref{fig:Relabel}~b) we show the itinerary of the master system $\mathcal{S}^\phi_m(\x_{m0}$) and in Figures \ref{fig:Relabel}~d), \ref{fig:Relabel}~f) and \ref{fig:Relabel}~h) the itinerary of the slave system by varying the initial condition. Notice that the itinerary of the master system and the three itineraries of the slave system after relabeling the atoms for different initial conditions present approximately the same itinerary at the same time intervals. Hence they achieve itinerary synchronization. Figure \eqref{fig:ErrorRelabel} shows three errors signals which were generated by the difference between the master itinerary and slave itineraries after relabeling the atoms for different initial conditions. These signals show spikes that correspond to  when the trajectory goes from one atom to other. Now the errors  $k$ in the three cases: for the initial condition $\chi_{so1} = (1.01, 0.48, -0.29)^{T}$ generates on average $k_{1}=0$, $\chi_{so2} = (3.5, 0.48, -0.29)^{T}$ generates $k_{2}=0$ and $\chi_{so3} = (5.3, 0.48, -0.29)^{T}$ generates $k_{3}=0$. 
These small peaks along the error signals, and  we achieve  itinerary synchronization in practice.

\begin{figure}[ht] 
\centering
\includegraphics[width=12cm,height=8cm]{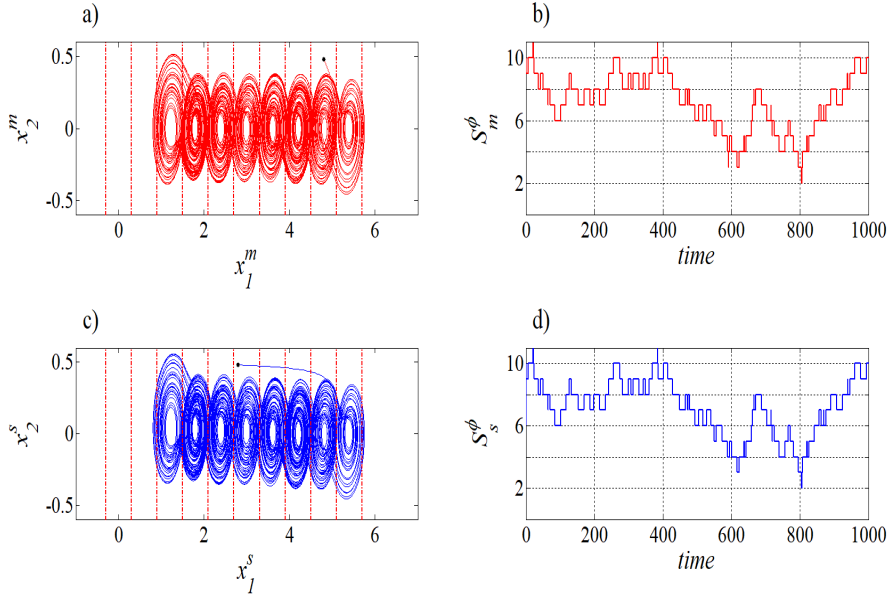} 
\caption{\small Projections of the master (red color) and slave (blue color) attractors onto the ($x_{1},x_{2}$) plane with $\eta_{m}=8$, $\eta_{s}=3$, $\Gamma = \{1,1,1\}$, coupling strength $c = 10$; initial conditions $\chi_{so} = ( 2.8,0.48, -0.29)^{T}$ and $\chi_{mo} = (4.8,0.48, -0.29)^{T}$for the slave and master systems, respectively.}
\label{fig:Fig4}
\end{figure}

{\bf Example 3:} 
Another example of coupled systems is given when the master scroll degree is greater than the slave scroll degree, for example, suppose that the master's scroll-degree is $\eta_{m} = 8$  and  the slave's scroll-degree is $\eta_{s} = 3$. Then, the pure-state-feedback signal for the master system is \eqref{eq:switching_law_8} and for the slave system is  \eqref{eq:switching_law_3}. The inner coupling matrix is $\Gamma = diag\{1,1,1\}$ and the coupling strength is $c=10$.  Figure \ref{fig:Fig4}~a) and c) shows the projections of the master and slave attractors given by \eqref{eq:MS} onto the ($x_{1}^m,x_{2}^m$) and ($x_{1}^s,x_{2}^s$) planes, respectively, generated with initial condition  $\chi_{mo}$ given above for the master system and $\chi_{so} = ( 2.8, 0.48, -0.29)^{T}$ for the slave system. Note that the slave system increases its scroll-degree to $\eta_{m} = 8$.  Figure \ref{fig:Fig4}~b) and d) shows the master and slave itineraries, respectively.   
This result is given by the following proposition~\ref{propMSS}.
 
\begin{prop}\label{propMSS}
Consider a master-slave system composed of  quasi-sym\-met\-ri\-cal $\eta$-PWL systems described by \eqref{eq:MS} and pure-state-feedback signals $\kappa_{m}(\x)$,  and $\kappa_{s}(\x)$ with  $\Gamma=diag\{1,1,1\}$. If the master-slave system presents com\-ple\-te synchronization this implies that the nodes present itinerary synchronization and the coupled pair of systems display the scroll-degree $\eta_m$ of the master system. 
\end{prop}
\begin{proof}
The master slave system is given by
\begin{equation}\label{eq:prop1}
\begin{array}{lll}
    \dot{\x}_{m}  & = A\x_{m}   + B_{\kappa_{m}(\x_{m})},  \\ 
    \dot{\x}_{s}  & =  A\x_{s}  + B_{\kappa_{s}(\x_{s})}  + c \Gamma( \x_{m} - \x_{s}).
\end{array} 
\end{equation}
When the system presents complete synchronization this implies that $|\x_m-\x_s|=0$. Defining the error between the master and slave systems as $e=\x_i-\x_j=(e^{x_1},e^{x_2},e^{x_3})^T$, where $e^{x_1}=x_{m1}-x_{s1}$, $e^{x_2}=x_{m2}-x_{s2}$ and $e^{x_3}=x_{m3}-x_{s3}$.
Thus the error system is given by
$$\dot{e}=Ae+B_{\kappa_{m}(\x_{m})}-B_{\kappa_{s}(\x_{s})}+c \Gamma e.$$
For complete synchronization we have that error  $lim_{t\to \infty}e= 0$ then $\lim_{t\to \infty}\dot{e}=0$. So we have that $B_{\kappa_{s}(\x_{s})}$ behaves in a similar manner to $B_{\kappa_{m}(\x_{m})}$, that is, $B_{\kappa_{s}(\x_{s})} \to B_{\kappa_{m}(\x_{m})}$. \\
The master-slave system displays itinerary synchronization. 
\end{proof}
 
\subsection{Dynamical Networks}

A dynamical network is composed of $N$ coupled dynamical systems called nodes \cite{Wang2003b}. Each node  is labeled by an index $i = 1,\ldots,N$ and described by a first ordinary differential equation system of the form $\dot{\x}_{i}(t) = f_{i}(\x_{i}(t))$, where $x_i(t) = (x_{i1}(t),\ldots,x_{in}(t))^{T} \in \R^{n}$ is the state vector and, $f_{i}: \R^{n} \rightarrow \R^{n}$ is  the vector field which describes the dynamical behavior of an $i$-th node when it is not connected to the network.  The coupling between neighboring nodes is assumed to be linear such that the state equation of the entire network is described by the following equations:
\begin{equation}\label{eq:dynamial_network}
\dot{\x}_{i}(t) = f_{i}(\x_{i}(t)) + c \sum_{j = 1}^{N} \Delta_{ij} \Gamma (\x_{j}(t) - \x_{i}(t)), \quad i = 1,\ldots,N,
\end{equation}

\noindent where $c$ is the uniform coupling strength between the nodes and the inner linking matrix $\Gamma  = diag\{r_1,\ldots,r_n \} \in \R^{n \times n}$ is described in \eqref{eq:MS}.  The constant matrix  $\Delta = \{ \Delta_{ij} \} \in \R^{N \times N}$ is named the coupling matrix whose elements are zero or one depending on which nodes are connected or not. Such matrix contains the entire information about the network configuration topology and it is constructed according to its link attributes.  In specific, if nodes are coupled with bidirectional links, then $\Delta$ is a symmetric matrix with the following entries:
if there is a connection between node $i$ and node $j$ (with $i \neq j$), then $\Delta_{ij} = \Delta_{ji} = 1$; otherwise $\Delta_{ij} = \Delta_{ji} = 0$. 

On the other hand, if the nodes are connected with unidirectional links, then $\Delta$ becomes a non-symmetric matrix  and its entries are defined as follows: 
  $\Delta_{ij} = 1$ (with $i \neq j$)  indicates the presence of an edge directed from  node $j$ to node $i$; and the entry $\Delta_{ij} = 0$ indicates that node $j$ is not connected to node $i$. 

For the dynamical network \eqref{eq:dynamial_network} with symmetrical coupling matrix, one of the most studied collective phenomena is  synchronization, which emerges when the dynamical behavior between nodes are correlated in-time (See \cite{Wang2003b} and references there in). 


\section{Ring topology network}
\begin{figure}[ht]
  \centering
        \begin{subfigure}[b]{0.35\textwidth}
          \includegraphics[width=\textwidth,, height=\textwidth]{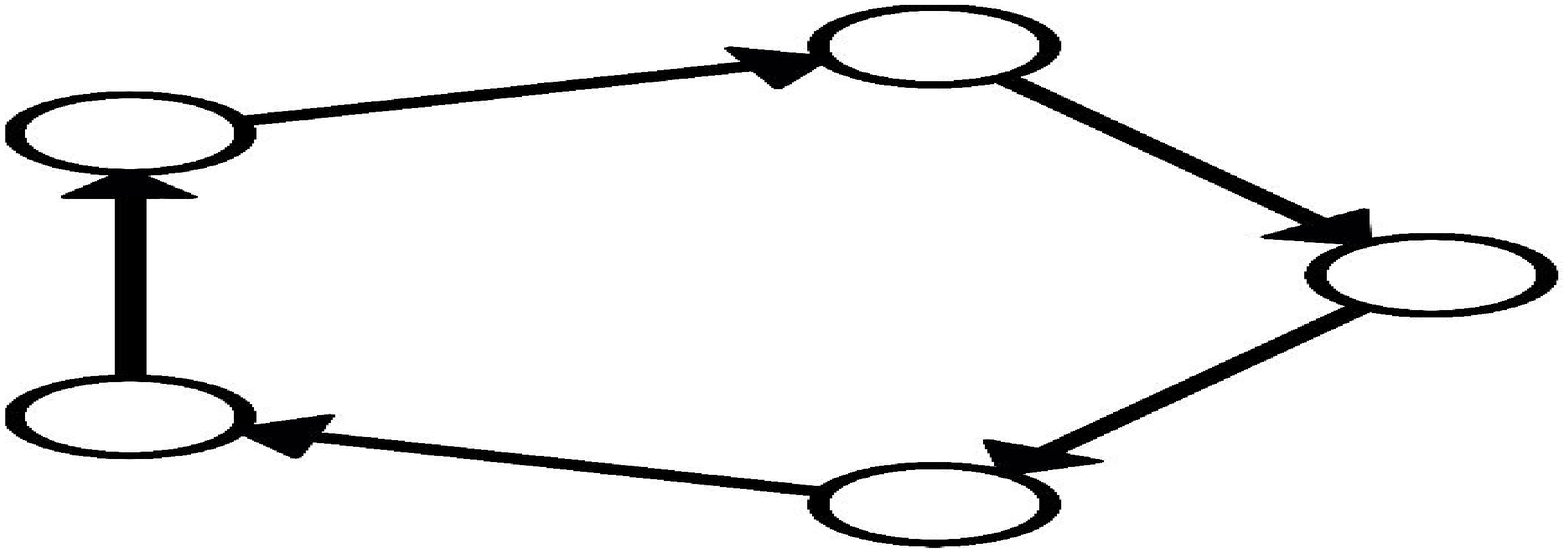}
         \caption{}
         \label{fig4a}
       \end{subfigure}
       ~
   \hspace{-1pt} 
   \begin{subfigure}[b]{0.45\textwidth} 
   \vspace{10pt}  
      \[
        \Delta = \begin{pmatrix}
           $0$ & $0$ & $0$ & $0$ & $1$  \\
	   $1$ & $0$ & $0$ & $0$ & $0$ \\
	   $0$ & $1$ & $0$ & $0$ & $0$ \\
	   $0$ & $0$ & $1$ & $0$ & $0$ \\
	   $0$ & $0$ & $0$ & $1$ & $0$ \\
       \end{pmatrix}
        \]
       \caption{}
       \label{fig4b} 
       \end{subfigure}
\caption{A network of $N = 5$ nodes coupled in a ring topology with uni-directional links. a) The network topology and b) the coupling matrix.}
\label{fig:Ring_Connection}
\end{figure}
We study the collective dynamics of $N$ coupled quasi-symmetrical $\eta$-PWL systems which are connected by unidirectional links in a ring topology, \textit{i.e.}, a network composed of an ensemble of master-slave systems coupled in a cascade configuration topology. In this context, a system defined in the node $i$ is a slave system of a system defined in the node $i-1$, and also plays the role of a master system for a system defined in the node $i+1$. Figure~\ref{fig:Ring_Connection}~(a)  shows a network with a ring topology and  \ref{fig:Ring_Connection}~(b) its corresponding coupling matrix $\Delta$. A network with such attributes is described by the following state equations:

\begin{equation}\label{eq:MS_ClosedRing}
\left\{
\begin{array}{l}
    \dot{\x}_{1}   = A\x_{1}   + B_{\kappa_{1}(\x_{1})}  + c\Gamma ( \x_{N} - \x_{1}), \\ 
    \dot{\x}_{2}   =  A\x_{2}  + B_{\kappa_{2}(\x_{2})}  + c\Gamma ( \x_{1} - \x_{2}), \\
    \dot{\x}_{3}   =  A\x_{3}  + B_{\kappa_{3}(\x_{3})}  + c\Gamma ( \x_{2} - \x_{3}),\\
     \hspace{5pt} \vdots           \hspace{50pt} \vdots  \\
    \dot{\x}_{N}    =  A\x_{N}  +  B_{\kappa_{N}(\x_{N})} + c\Gamma ( \x_{N-1} - \x_{N}),    
\end{array} 
\right.
\end{equation}

\noindent where $\x_i$, $i = 1, 2,\ldots, N$, denotes the state vector of each node.  Notice that the system \eqref{eq:MS_ClosedRing} is a dynamical network where each node differs only in the constant vector $B_{\kappa_{i}(\cdot)}$. In this context, we propose the following definition of a network of nearly identical nodes:

\begin{defi} 
A network of nearly identical nodes is a network composed of nodes with dynamics given by quasi-symmetrical $\eta$-PWL systems, {\it i.e.}, $A_{i} = A_{j} = A$, $\eta_{i} \neq \eta_{j}$ and $\kappa_{i}(\cdot) \neq  \kappa_{j}(\cdot)$ $\forall i,j =1,2,\ldots,N$ whose state equation is written as follows:
\begin{equation}\label{eq:nearest-identical-network}
\dot{\x}_{i} =  A\x_{i} + B_{\kappa_{i}(\x_{i})}  + c \sum_{j = 1}^{N} \Delta_{ij} \Gamma ( \x_{j} - \x_{i}), \quad i = 1,\ldots,N.
\end{equation}
\end{defi}

Note that \eqref{eq:nearest-identical-network} corresponds to a dynamical network with a  configuration topology given by the coupling matrix $\Delta = \{\Delta_{ij}\} \in \R^{N \times N}$. In particular, for a ring topology (Figure~\ref{fig:Ring_Connection}), the equation \eqref{eq:nearest-identical-network} becomes the equation \eqref{eq:MS_ClosedRing}.

We now study the dynamics of a nearly identical network \eqref{eq:nearest-identical-network} assuming that the coupling matrix  corresponds to a network with a ring topology and with unidirectional links. We analyze  the conditions under which this nearly identical network  achieves itinerary synchronization. We then consider  the case in which the network \eqref{eq:nearest-identical-network}  is attacked via link deletion. 


\subsection{Node's dynamics}

Since the dynamics of a single node is governed by an UDS system, we know that the linear operator is diagonalizable \textit{i.e.} exist a matrix $\Phi \in \R^{3\times 3}$ such that $\Lambda = \Phi A \Phi^{T}$ with $\Lambda = diag\{\lambda_{1},\lambda_{2},\lambda_{3}$\} and $\Phi\Phi^{T} = \mathit{I}$. By introducing the change of variable $z = \Phi \x$, we rewrite the equation \eqref{eq:nearest-identical-network} as follows: 
\begin{equation}\label{eq:nearest_net_variable_change}
\dot{z}_{i} =  \Lambda z_{i} + \hat{B}_{\kappa_{i}(z_{i})}  + c \sum_{q = 1}^{N} \Delta_{iq} \Gamma z_{iq}  , \quad i = 1,\ldots,N,
\end{equation}

\noindent where $\hat{B}_{\kappa_{i}(z_{i})} = \Phi B_{\kappa_{i}(\x_{i})}$ and $z_{iq} = z_{q} - z_{i}$. The solution of \eqref{eq:nearest_net_variable_change} is:
\begin{equation}\label{eq:net_solution}
   z_{i}(t) = e^{\Lambda t} z_{i}(0)  + \int_{0}^{t}  e^{\Lambda (t-\tau)} \hat{B}_{\kappa_{i}(z_{i})} d\tau +   c \sum_{q = 1}^{N} \Delta_{iq} \Gamma \int_{0}^{t} e^{\Lambda (t-\tau)}  z_{iq}(\tau)  d\tau,
\end{equation}

\noindent where  $z_{i}(0)$ is the initial condition of the i-th node in the new state variable. Note that  \eqref{eq:net_solution} is given by using the Peano-Baker series with $e^{\Lambda t}$ as the state transition matrix since $\Lambda$ is a constant matrix. 

The difference between the state vector of the i-th and j-th nodes is:
\begin{equation}\label{eq:difference}
   z_{i}(t) - z_{j}(t) = e^{\Lambda t} z^{0}_{ij}  + \int_{0}^{t}  e^{\Lambda (t-\tau)} \hat{B}_{ij} d\tau +   c \sum_{q = 1}^{N}  \Gamma \int_{0}^{t} e^{\Lambda (t-\tau)}  \hat{z}^{q}_{ij}(\tau)  d\tau,
\end{equation}
 
\noindent where $z^{0}_{ij} = z_{i}(0) - z_{j}(0)$, $\hat{B}_{ij} = \hat{B}_{\kappa_{i}(z_{i})} - \hat{B}_{\kappa_{j}(z_{j})} $ and  $\hat{z}^{q}_{ij}(\tau) = \Delta_{iq}z_{iq}(\tau)  - \Delta_{jq}z_{jq}(\tau)$. Then, by assuming that $|| z_{i}(0) - z_{j}(0) || \leq \delta_{0}$ for all $i,j = 1, \ldots, N$ and using the triangle inequality we get:

\begin{equation}\label{eq:inequality_cond}
   ||z_{i}(t) - z_{j}(t)|| \leq ||e^{\Lambda t}|| \delta_{0}  +  I_{1}(t)  +    I_{2}(t),
\end{equation}

\noindent where
\begin{equation}\label{eq:I1}
  I_{1}(t) = \int_{0}^{t}  || e^{\Lambda (t-\tau)} \hat{B}_{ij} || d\tau, \quad \text{and} 
  \quad  I_{2}(t) = c \sum_{q = 1}^{N}  \int_{0}^{t} \Gamma  || e^{\Lambda (t-\tau)}  \hat{z}^{q}_{ij}(\tau) || d\tau .
\end{equation}

Both integrals satisfy $I_{1}(0) = I_{2}(0) = 0$. Then, in a practical sense  complete synchronization is achieved if there exists a time instant $T$ such that $I_{1}(t) \leq \epsilon_{1}$ and $I_{2}(t) \leq \epsilon_{2}$ for all $t \geq T$ and for all $i,j = 1,\ldots,N$. In this case the error bound is given by $\epsilon = ||e^{\Lambda T}|| \delta_{0} + \epsilon_{1} +\epsilon_{2}$. Note that if the node trajectories approach each other, then $\hat{B}_{ij} = 0$. Itinerary synchronization occurs since both signals generate the same symbolic dynamics. The error bound depends mainly on $I_{2}(t)$, and in particular, in the form of the inner coupling matrix $\Gamma$.

In a master-slave system, when the scroll-degree of the slave system is greater than the scroll-degree of the master then multistability phenomenon appears and complete synchronization implies itinerary synchronization as it was shown in proposition~\ref{propMSS}. However the converse is not true, {\it i.e.}, itinerary synchronization does not imply complete synchronization if there is multistability (different basins of attraction $\Omega_1, \Omega_2,\ldots,\Omega_k$, with $2\leq k\in\Z$). For example, when the initial conditions of the slave system belong to different basins of attraction, this leads to
$$\lim_{t\to \infty}|S^\phi_m(i,\chi_{m0})-S_s^\phi(i,\chi_{s0})|=k.$$
With $k\neq 0$ and  complete synchronization is lost but itinerary synchronization persists after relabeling.


\section{Different inner coupling matrix $\Gamma$ }\label{InnerCouplingSec}


\subsection{Dynamics in a ring topology}

\begin{figure}[t]
\centering
\includegraphics[width=11.5cm,height=10cm]{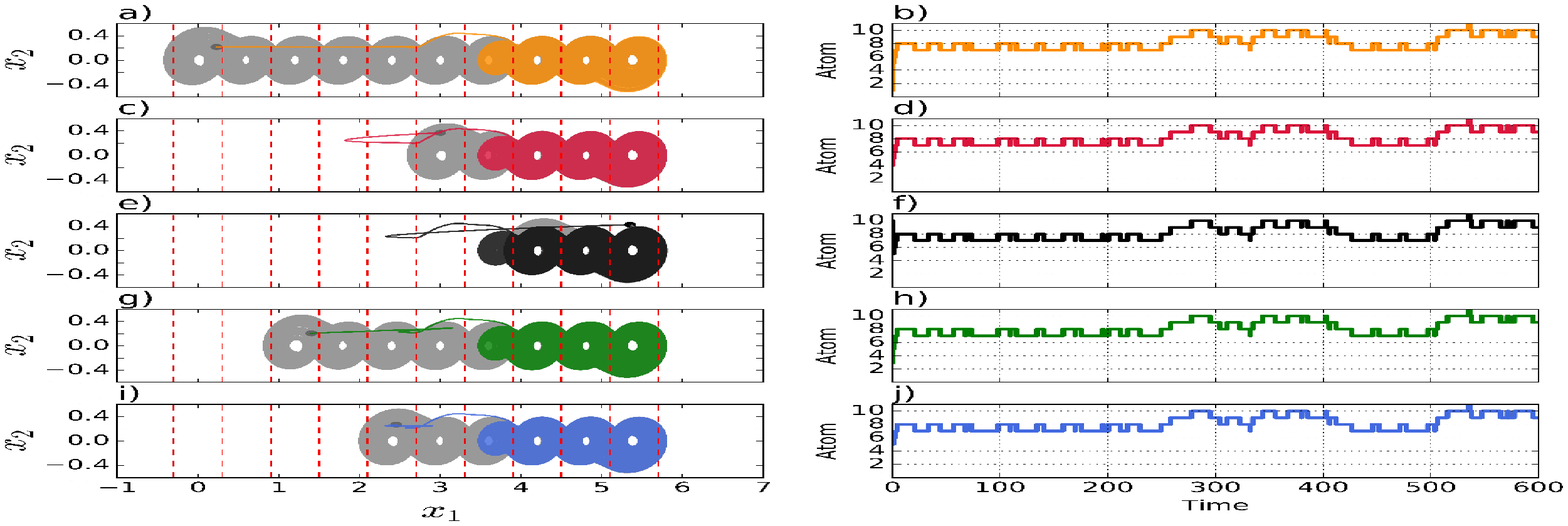}
\caption{\small Dynamics of a  nearly identical network \eqref{eq:nearest-identical-network} with coupling strength $c=10$ and $\Gamma = diag\{1,1,1\}$; the scroll-degree and initial  condition for each node are given in Table \eqref{table:scroll_degree_config}: a); c); e); g); i); The projections of the attractors onto the plane $(x_1 ,x_2)$ of the node 1,2,3,4 and 5 respectively (the gray line represent the projection of the trajectory of the node without connection); and b); d); f); h); j) its itinerary.}
\label{fig:Example_4}
\end{figure}

We consider a ring network with five nodes, {\it i.e.}, $N = 5$ nearly identical nod\-es
 described in \eqref{eq:nearest-identical-network} and coupled in a ring topology. We assume that each node's dynamic is described by the same linear operator $A$ (\textit{i.e} they are quasi-symmetrical) and the set of constant vectors $\textbf{B} = \{B_{1},B_{2},\ldots,B_{10}\}$  are those given by \eqref{eq:example1}.  Further, for each node we select the scroll-degree ($\eta_{i}$) and its corresponding initial condition according to Table \eqref{table:scroll_degree_config}. 

\begin{table}[ht]
\centering
\begin{tabular}{ccl}
Node's label & Scroll-degree & Initial condition  \\
\midrule
\rowcolor{black!20}  1 & 10 & $( 0.227, -0.216, -0.359)^{T}$\\
                                 2 & 5  &  $( 3.014, -0.371, -0.271)^{T}$ \\
\rowcolor{black!20}  3 & 3   & $( 5.349, -0.424, -0.279)^{T}$\\
			       4 &  8    &  $(1.402, -0.205, -0.316)^{T}$\\
\rowcolor{black!20}  5 &  6     &  $(2.452, -0.266, -0.308)^{T}$\\
\bottomrule
\end{tabular}
\caption{The scroll-degree ($\eta_{i}$) and its corresponding initial condition for each node in the nearly identical network coupled in a ring topology for Examples 4 and 5.}  
\label{table:scroll_degree_config}
\end{table}

The pure-state-feedback signal for the first node with scroll-degree $\eta_{1} = 10$ is given by  \eqref{eq:switching_law_10}; for the third and fourth nodes with scroll-degree $\eta_{3} = 3$ and $\eta_{4} = 8$ are given by  \eqref{eq:switching_law_3} and \eqref{eq:switching_law_8} respectively. For the second node with scroll degree $\eta_{2} = 5$ the pure-state-feedback signal is given as follows: 
\begin{equation}\label{eq:switching_law_5}
\kappa_{5}(\x) =  \left\{
\begin{array}{lll}
    1,   &  \text{if}   & \x \in P_{10} = \{\x \in \R^{3}: x_{1} \geq 5.1\}; \\
    2,   &  \text{if}   & \x \in P_{9} = \{\x \in \R^{3}: 4.5 \leq x_{1} < 5.1\}; \\
    3,  &  \text{if}   & \x \in P_{8} = \{\x \in \R^{3}: 3.9 \leq x_{1} < 4.5\};   \\
    4,   &  \text{if}   & \x \in P_{7} = \{\x \in \R^{3}: 3.3 \leq x_{1} < 3.9 \};   \\
    5,   &  \text{if}   & \x \in P_{6} = \{\x \in \R^{3}:  x_{1} < 3.3 \} . \\
\end{array} 
\right.
\end{equation}

\noindent And for the fifth node with scroll degree $\eta_{5} = 6$ is
\begin{equation}\label{eq:switching_law_6}
\kappa_{6}(x) =  \left\{
\begin{array}{lll}
    1, &  \text{if}   & \x \in P_{10} = \{\x \in \R^{3}: x_{1} \geq 5.1\}; \\
    2, &  \text{if}   & \x \in P_{9} = \{\x \in \R^{3}: 4.5 \leq x_{1} < 5.1\};  \\
    3, &  \text{if}   & \x \in P_{8} = \{\x \in \R^{3}: 3.9 \leq x_{1} < 4.5\};   \\
    4, &  \text{if}   & \x \in P_{7} = \{\x \in \R^{3}: 3.3 \leq x_{1} < 3.9 \};   \\
    5, &  \text{if}   & \x \in P_{6} = \{\x \in \R^{3}: 2.7 \leq x_{1} < 3.3 \};  \\
    6, &  \text{if}   & \x \in P_{5} = \{\x \in \R^{3}:  x_{1} < 2.7 \} . \\
\end{array} 
\right.
\end{equation}

\begin{figure}[ht]
\centering
\includegraphics[width=11.5cm,height=10cm]{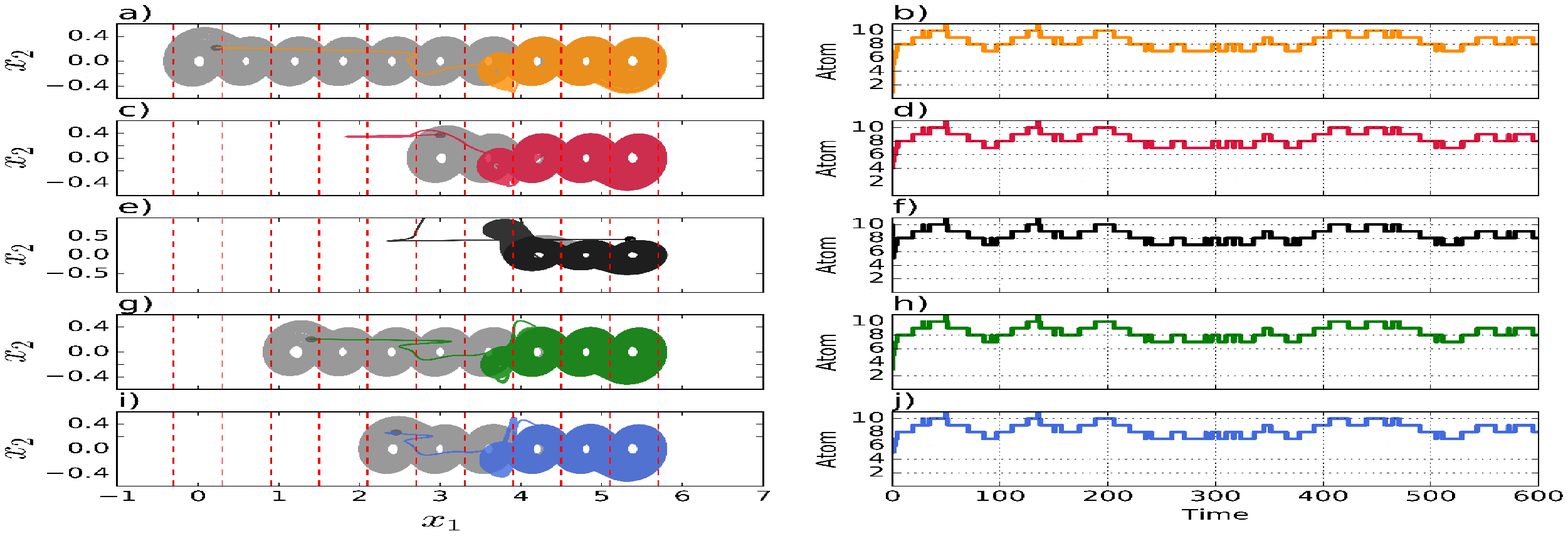}
\caption{ \small Dynamics of a  nearly identical network \eqref{eq:nearest-identical-network} with coupling strength $c=10$ and $\Gamma = diag\{1,0,0\}$; the scroll-degree and initial  condition for each node are given in Table \eqref{table:scroll_degree_config}: a); c); e); g); i); The projections of the attractors onto the plane $(x_1 ,x_2)$ of the node 1,2,3,4 and 5 respectively (the grey line represent the projection of the  trajectory of the node without connection); and b); d); f); h); j) its itinerary.}
\label{fig:Example_4a}
\end{figure}
{\bf Example 4:} For the nearly identical network described above, we assume that the  coupling strength is $c = 10$ and the inner coupling matrix is $\Gamma = diag\{1,1,1\} \in \R^{3}$. We solve numerically the nearly 
identical network \eqref{eq:nearest-identical-network} with the scroll-degree and initial condition given in Table \eqref{table:scroll_degree_config} and using a Runge-Kutta method with 10000 time iterations and step size $h=0.01$. In the first column of the Figures~\ref{fig:Example_4} we  shown the projections of the attractors onto the plane $(x_1,x_2)$ and in the right column we display its corresponding itinerary, {\it i.e.}, the time elapsed that the trajectory of each node spends in a given atom. Note that independently of wherever the initial conditions are defined, the trajectories of the nodes converge to the atoms 10, 9, 8 and 7 in a attractor similar to the node with the smallest scroll-degree. For this configuration with the inner connection $\Gamma= diag\{1,1,1\}$ the network achieves both complete and itinerary synchronization. In the next example we show a case in which complete synchronization is not achieved but the nodes are itinerary synchronized.
\begin{figure}[ht]
\centering
\begin{subfigure}[b]{0.3\textwidth}
\includegraphics[width=0.9\textwidth]{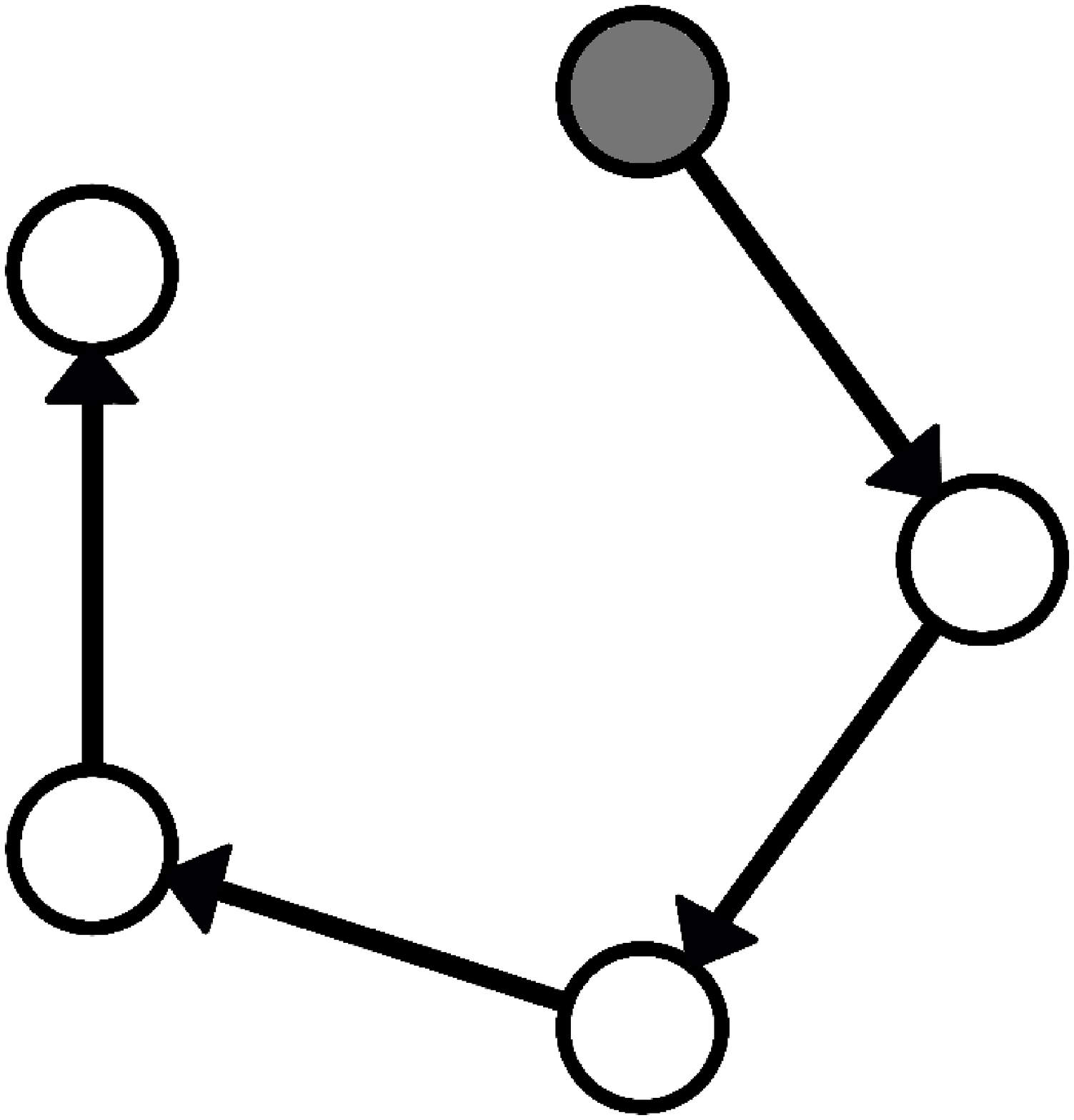}
\caption{}
\label{fig:Open_Ring}
\end{subfigure}
~
\hspace{-1pt} 
\begin{subfigure}[b]{0.5\textwidth}  
\vspace{10pt}  
\[
\Delta = \begin{pmatrix}
$0$ & $0$ & $0$ & $0$ & $0$  \\
$1$ & $0$ & $0$ & $0$ & $0$ \\
$0$ & $1$ & $0$ & $0$ & $0$ \\
$0$ & $0$ & $1$ & $0$ & $0$ \\
$0$ & $0$ & $0$ & $1$ & $0$ \\
\end{pmatrix}
\]
\caption{}
\label{fig:Adj_Open_Ring}
\end{subfigure}
\caption{A network of $N = 5$ nodes coupled in a open ring topology with directional links. (a) The network topology and (b) the coupling matrix.}
\label{fig:Open_Ring_Connection}
\end{figure}
\begin{figure}[ht]
\centering
\includegraphics[width=11.5cm,height=10cm]{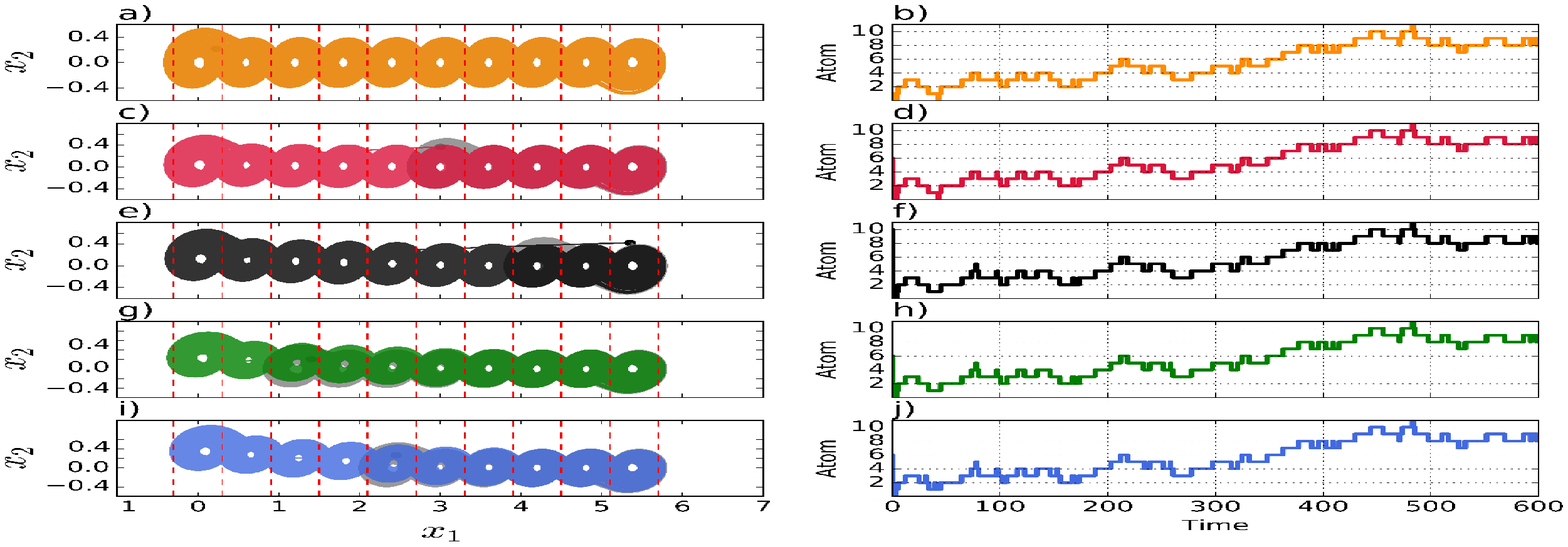}
\caption{  \small a),~c),~e),~g),~i): The projections of the attractors onto the plane $(x_1 ,x_2)$ of the nodes of a  nearly identical network \eqref{eq:nearest-identical-network} in an attacked ring topology  with coupling strength $c=10$, $\Gamma = diag\{1,1,1\}$ and where the first node has scroll-degree $\eta_{i}=10$.  b),~d),~f),~h),~j): The itinerary of each node.}
\label{fig:Example_8} 
\end{figure}

{\bf Example 5:} The dynamics of the  network composed of $N$ quasi-sym\-met\-ri\-cal $\eta$-PWL systems described above can display several behaviors depending on the inner coupling matrix $\Gamma$. The collective dynamics is affected when we suppress some variable state in the inner connection. For example, in Figure~\ref{fig:Example_4a} when we suppress two state variables  from the inner coupling matrix $\Gamma = diag\{1,0,0\}$, a deformation of the scroll attractor is achieved specially over the node with the smallest node-degree (in this case, for the node with scroll-degree 3). However the nodes still share the same itinerary. That is, we observe that complete synchronization is not achieved, however, the trajectories visit the same atoms during the same time intervals.

In the next subsection we consider the case in which this network is attacked via link deletion.

\subsection{Link attack in the ring topology}

In this subsection we present numerical results for the case in which the network is attacked by removing a preselected link. This attack transforms the network topology from a ring configuration to  a chain (open ring) configuration as we illustrate in Figure~\ref{fig:Open_Ring_Connection}; where we also shown the corresponding coupling matrix that describes this network.

After the attack, the node painting in black in Figure~\ref{fig:Open_Ring_Connection}~(a), which we call the leader node, plays the role of the master system for the rest of the nodes. The second node is the slave system for the leader node, but it is also the master system for the third node, and so on. The idea is to  explore if such a node governs or not the collective dynamics of the rest of the nodes. In this work we assume that the scroll-degree of the master node corresponds to the largest or the smallest scroll-degree. Specifically we  consider two examples: the first node has scroll-degree ten or  three.

\subsubsection{Master system with maximum scroll-degree}
{\bf Example 6:}  Figure \ref{fig:Example_8} shows the projections onto the plane $(x_1,x_2)$ of the attractors generated in each node  by the nearly identical network \eqref{eq:nearest-identical-network}  with an open ring configuration. For this example we assume that  the first node has scroll-degree $\eta_{1}=10$, and the nodes are connected with coupling strength  $c=10$ and  inner coupling matrix $\Gamma= diag\{1,1,1\}$. The node's scroll-degree and its corresponding  initial condition are given in Table \eqref{table:scroll_degree_config}. All the nodes  imitate the dynamics of the master system and change their dynamics to attain the same scroll-degree. In this context, the scroll-degree of the leader node dominates  and itinerary synchronization is achieved. 

\begin{figure}[ht]
\centering
\includegraphics[width=11.5cm,height=10cm]{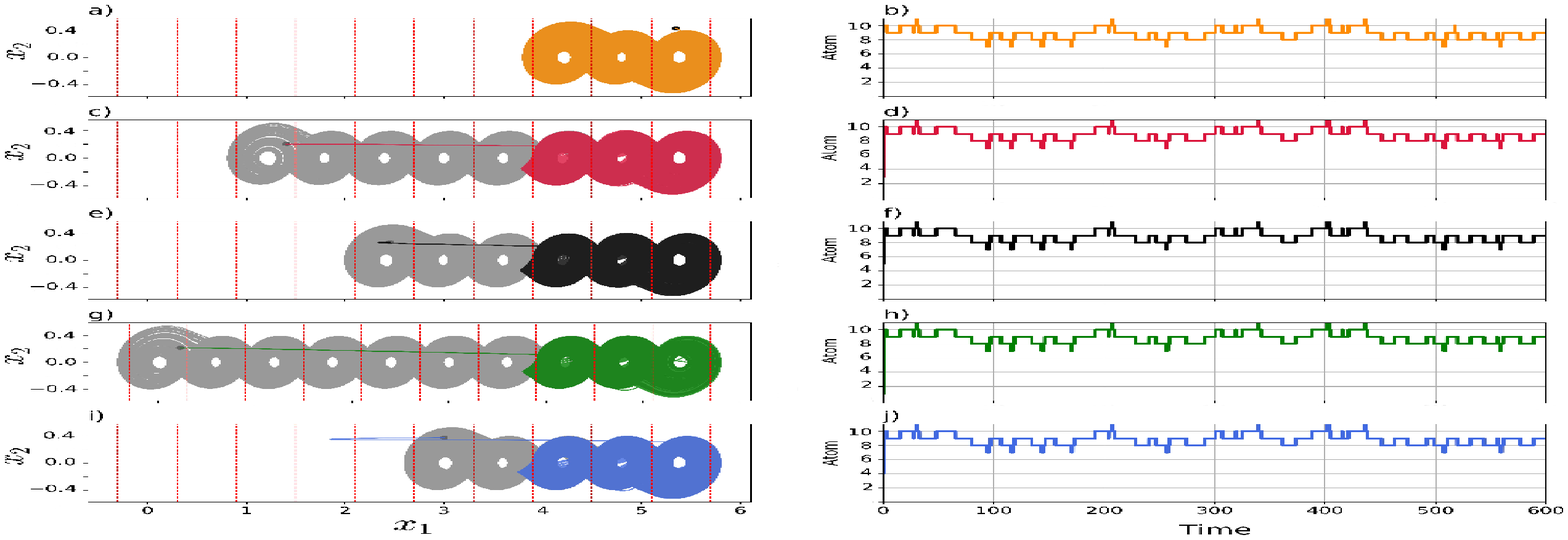}
\caption{  \small a),~c),~e),~g),~i): The projections of the attractors onto the plane $(x_1 ,x_2)$ of a  nearly identical network \eqref{eq:nearest-identical-network} in a attacked ring topology  with coupling strength $c=10$, $\Gamma = diag\{1,1,1\}$ and where the first node has scroll-degree $\eta_{i}=3$. b),~d),~f),~h),~j): The itinerary of each node.}
\label{fig:Example_9} 
\end{figure}

\subsubsection{Master system with minimum scroll-degree}
{\bf Example 7:}  Now we assume that after removing the link, the first node has scroll-degree $\eta_{1} = 3$, and the rest of the nodes have the scroll-degree and initial condition given in Table \eqref{table:scroll_degree_config}. As before, we select a coupling strength $c=10$ and $\Gamma = diag\{1,1,1\}$. In Figure \eqref{fig:Example_9}  we observe that  all the nodes reduce their scroll-degree to three \textit{i.e.} the nodes adopt the scroll-degree of the first node. Furthermore, the rest of the nodes achieve both complete and itinerary synchronization for the set of given initial conditions.


\section{Conclusions}

We investigate the collective dynamics of a network composed of PWL-system where the number of scroll-attractors in each node differs. We named such a type of system a nearly-identical network and we used the term  scroll-degree  of a node to denote the  number of scroll attractors in an individual node. Furthermore, we assumed that the network topology corresponds to a ring configuration such that the entire system can be seen as an ensemble of master-slave systems connected by directional links.

Our   numerical results show that itinerary synchronization can be achieved in this setting. Furthermore, we found that the node with the smallest scroll-degree governs the collective itinerary of the network, i.e., the dominant node in a ring configuration network is that with smallest scroll-degree. Additionally, we introduced link attacks to the network and transformed the network topology to a open ring configuration. We explored two possible scenarios after attack: the first node has the a) the largest or b) the smallest scroll-degree. In both scenarios we observed that scroll-degree of the leader node dominates. Furthermore, in the above two scenarios, itinerary synchronization is  achieved. 


\section*{Acknowledgment}
 
E. Campos-Cant\'on acknowledges CONACYT for the financial support for sabbatical year. 
M. Nicol thanks the NSF for partial support on NSF-DMS Grant 1600780.


\end{document}